\documentclass[twocolumn, a4paper, superscriptaddress,nofootinbib, accepted=2020-08-07, hyperref]{quantumarticle}
\pdfoutput=1
\usepackage{enumitem}
\usepackage{graphicx}
\usepackage{bm}
\usepackage{amsmath}
\usepackage{amssymb}
\usepackage{underscore}
\usepackage{color}
\usepackage{float}
\usepackage[caption = false]{subfig}
\usepackage{amsmath,amsthm,amsfonts,amssymb,amscd}

\usepackage{hyperref}
\hypersetup{
    colorlinks=true,
    linkcolor=blue,
    filecolor=magenta,      
    urlcolor=violet,
}

\def\squareforqed{\hbox{\rlap{$\sqcap$}$\sqcup$}}
\def\qed{\ifmmode\squareforqed\else{\unskip\nobreak\hfil
\penalty50\hskip1em\null\nobreak\hfil\squareforqed
\parfillskip=0pt\finalhyphendemerits=0\endgraf}\fi}

\def\duzomniejsze{<\kern-.7mm<}
\def\duzowieksze{>\kern-.7mm>}

\def\textbf#1{{\bf #1}}
\def\beq{\begin{equation}}
\def\eeq{\end{equation}}
\def\be{\begin{equation}}
\def\ee{\end{equation}}
\def\ben{\begin{eqnarray}}
\def\een{\end{eqnarray}}
\def\beqa{\begin{eqnarray}}
\def\eeqa{\end{eqnarray}}
\def\eea{\end{array}}
\def\bea{\begin{array}}

\newcommand{\bei}{\begin{itemize}}
\newcommand{\eei}{\end{itemize}}
\newcommand{\bee}{\begin{enumerate}}
\newcommand{\eee}{\end{enumerate}}
\newcommand{\nc}{\newcommand}
\nc{\1}{{\openone}}

\def\>{\rangle}
\def\<{\langle}

\newtheorem{lemma}{Lemma}

\newtheorem{theorem}{Theorem}
\newtheorem{dfn}{Definition}

\newtheorem*{definition}{Definition}

\def\bed{\begin{definition}}
\def\eed{\end{definition}}
\def\bel{\begin{lemma}}
\def\eel{\end{lemma}}
\def\bet{\begin{theorem}}
\def\eet{\end{theorem}}
\def\be{\begin{equation}}
\def\ee{\end{equation}}

\allowdisplaybreaks

\usepackage{mathtools}

\begin{document}


\title{Gadget structures in proofs of the Kochen-Specker theorem}
\author{Ravishankar Ramanathan}
\affiliation{Department of Computer Science, The University of Hong Kong, Pokfulam Road, Hong Kong}

\author{Monika Rosicka}
\affiliation{Institute of Theoretical Physics and Astrophysics and the National Quantum Information Centre, Faculty of Mathematics, Physics and Informatics, 
	University of Gdansk, 80-308 Gdansk, Poland.}

\author{Karol Horodecki}
\affiliation{Institute of Informatics Faculty of Mathematics, Physics and Informatics,
		University of Gdansk, 80-308 Gdansk, Poland}
\affiliation{International Centre for Theory of Quantum Technologies,
		University of Gdansk, Wita Stwosza 63, 80-308 Gdansk, Poland}	

\author{Stefano Pironio}
\affiliation{Laboratoire d'Information Quantique, Universit\'{e} Libre de Bruxelles, Belgium}

\author{Micha{\l} Horodecki}
\affiliation{International Centre for Theory of Quantum Technologies, University of Gda\'{n}sk, Wita Stwosza 63, 80-308 Gda\'{n}sk, Poland}
\affiliation{Institute of Theoretical Physics and Astrophysics and the National Quantum Information Centre, Faculty of Mathematics, Physics and Informatics, 
	University of Gdansk, 80-308 Gdansk, Poland.}

\author{Pawe{\l} Horodecki}
\affiliation{International Centre for Theory of Quantum Technologies, University of Gda\'{n}sk, Wita Stwosza 63, 80-308 Gda\'{n}sk, Poland}
\affiliation{Faculty of Applied Physics and Mathematics, National Quantum Information Centre, Gda\'{n}sk University of Technology, Gabriela Narutowicza 11/12, 80-233 Gda\'{n}sk, Poland} 

\begin{abstract}
	The Kochen-Specker theorem is a fundamental result in quantum foundations that has spawned massive interest since its inception. We show that within every Kochen-Specker graph, there exist interesting subgraphs which we term $01$-gadgets, that capture the essential contradiction necessary to prove the Kochen-Specker theorem, i.e,. every Kochen-Specker graph contains a $01$-gadget and from every $01$-gadget one can construct a proof of the Kochen-Specker theorem. Moreover, we show that the $01$-gadgets form a fundamental primitive that can be used to formulate state-independent and state-dependent statistical Kochen-Specker arguments as well as to give simple constructive proofs of an ``extended'' Kochen-Specker theorem first considered by Pitowsky in \cite{Pitowsky}.  
\end{abstract}

\maketitle


\section{Introduction}
According to the quantum formalism, a projective measurement $M$ is described by a set $M=\{V_1,\ldots,V_m\}$ of projectors $V_i$ in a complex Hilbert space, that are orthogonal, $V_iV_j=\delta_{ij} V_i$, and sum to the identity, $\sum_i V_i=I$. Each $V_i$ corresponds to a possible outcome $i$ of the measurement $M$ and determines the probability of this outcome when measuring a state $|\psi\rangle$ through the formula $\text{Pr}_{\psi}(i\mid M)=\langle\psi|V_i|\psi\rangle$. 

If two physically distinct measurements $M=\{V_1,\ldots,V_m\}$ and $M'=\{V'_1,\ldots,V'_{m'}\}$ share a common projector, i.e., $V_i=V'_{i'}=V$ for some outcome $i$ of $M$ and $i'$ of $M'$, it then follows that
\begin{equation}\label{eq:meas}
\text{Pr}_{\psi}(i\mid M)=\text{Pr}_{\psi}(i'\mid M')=\langle\psi|V|\psi\rangle\,.
\end{equation}
In other words, though quantum measurements are defined by \emph{sets} of projectors, the outcome probabilities of these measurements are determined by the \emph{individual} projectors alone, independently of the broader set -- or the \emph{context} -- to which they belong. We say that the probability assignment is \emph{non-contextual}.

The Kocken-Specker (KS) theorem \cite{KS} is a cornerstone result in the foundations of quantum mechanics, establishing that, in Hilbert spaces of dimension greater than two, it is not possible to find a \emph{deterministic} outcome assignment that is non-contextual. Deterministic means that all outcome probabilities should take only the values $0$ or $1$. Non-contextual means, as above, that these probabilities are not directly assigned to the measurements themselves, but to the individual projectors from which they are composed, independently of the context to which the projectors belong. More formally, the KS theorem establishes that it is not possible to find a rule $f$ such that
\begin{equation}\label{eq:measlambda}
\text{Pr}_{f}(i\mid M)=\text{Pr}_{f}(i'\mid M')=f(V)\in\{0,1\}\,,
\end{equation}
which would provide a deterministic analogue of a quantum state.

The most common way to prove the KS theorem involves a set $\mathcal{S}=\{V_1,\ldots,V_n\}$ of rank-one projectors in a complex Hilbert space.
We can represent these projectors by the vectors (strictly speaking, the rays) onto which they project and thus view $\mathcal{S}$ as a set of vectors $\mathcal{S}=\{|v_1\rangle,\ldots,|v_n\rangle\}\subset\mathbb{C}^d$. Consider an assignment $f:\mathcal{S}\rightarrow \{0,1\}$ that associates to each $|v_i\rangle$ in $\mathcal{S}$ a probability $f(|v_i\rangle)\in\{0,1\}$.
To interpret the $f(|v_i\rangle)$ as valid measurement outcome probabilities, they should satisfy the two following conditions:
\begin{flalign}\label{eq:01rule}
\begin{minipage}{.42\textwidth}
\begin{itemize}
		\item $\sum_{|v\rangle\in \mathcal{O}} f(|v\rangle)\leq 1$ for every set $\mathcal{O}\subseteq \mathcal{S}$ of mutually orthogonal vectors;	
		\item $\sum_{|v\rangle\in \mathcal{B}} f(|v\rangle)=1$ for every set $\mathcal{B}\subseteq\mathcal{S}$ of $d$ mutually orthogonal vectors.
		\end{itemize}
\end{minipage}&&
\end{flalign}
The first condition is required because if a set of vectors are mutually orthogonal, they may be part of the same measurement, but then their corresponding probabilities must sum at most to 1. The second condition follows from the fact that if $d$ vectors are mutually orthogonal in $\mathbb{C}^d$, they form a complete basis, and then their corresponding probabilities must exactly sum to one. Note that the first condition implies in particular that any two  vectors $|v_1\rangle$ and $|v_2\rangle$ in $\mathcal{S}$ that are orthogonal cannot both be assigned the value 1 by $f$.

We call any assignment $f : \mathcal{S} \rightarrow \{0,1\}$ satisfying the above two conditions, a $\{0,1\}$-coloring of $\mathcal{S}$. The Kocken-Specker theorem states that if $d\geq 3$, there exist sets of vectors that are not $\{0,1\}$-colorable, thus establishing the impossibility of a non-contextual deterministic outcome assignment. We call such $\{0,1\}$-uncolorable sets, KS sets. In their original proof, Kochen and Specker describe a set $\mathcal{S}$ of 117 vectors in $\mathbb{C}^d$ dimension $d=3$ \cite{KS}. The minimal KS set contains 18 vectors in dimension $d=4$ \cite{CEG96, Cab08}. 	


In this paper, we identify within KS sets interesting subsets which we term $01$-gadgets. Such $01$-gadgets are $\{0,1\}$-colorable and thus do not represent by themselves KS sets. However, they do not admit arbitrary $\{0,1\}$-coloring: in any $\{0,1\}$-coloring of a $01$-gadget, there exist two non-orthogonal vectors $|v_1\rangle$ and $|v_2\rangle$ that cannot both be assigned the color 1. We show that such $01$-gadgets form the essence of the KS contradiction, in the sense that every KS set contains a $01$-gadget and from every $01$-gadget one can construct a KS set.

Besides being useful in the construction of KS sets, we show that $01$-gadgets also form a fundamental primitive in constructing statistical KS arguments \`a la Clifton \cite{Clifton93} and state-independent non-contextuality inequalities as introduced in \cite{YO12}. Moreover, we show that an ``extended" Kochen-Specker theorem considered by Pitowsky \cite{Pitowsky} and Abbott et al. \cite{ACS15, ACCS12} can be easily proven using an extension of the notion of $01$-gadgets. We give simple constructive proofs of these different results. 

Certain $01$-gadgets have already been studied previously in the literature, as they possess other interesting properties. In particular, $01$-gadgets were also used in \cite{Arends09} to show that the problem of checking whether certain families of graphs (which represent natural candidates for KS sets) are $\{0,1\}$-colorable is NP-complete, a result which we refine in the present paper. Specific $01$-gadgets have already been studied in the literature, for instance as 'definite prediction sets' in \cite{CA96} and recently as 'true-implies-false sets' in \cite{APSS18} where also minimal constructions in several dimensions were explored. A first method to produce different $01$-gadgets was also shown in \cite{CG95}. 


This paper is organized as follows. In section~\ref{sec:prel}, we introduce some notation and elementary concepts, in particular the representation of KS sets as graphs. In section~\ref{sec:gadget}, we define the notion of $01$-gadgets and establish their relation to KS sets. In section~\ref{sec:constr}, we give several constructions of $01$-gadgets and associated KS sets. In section~\ref{sec:real}, we show how $01$-gadgets can be used to construct statistical KS arguments. In section~\ref{sec:ext}, we also show a simple constructive proof of the extended Kochen-Specker theorem of Pitowsky \cite{Pitowsky} and Abbott et al. \cite{ACCS12} using a notion of extended $01$-gadgets which we introduce.  In section~\ref{sec:compl}, we show that $01$-gadgets can be used to establish the NP-completeness of $\{0,1\}$-coloring of the family of graphs relevant for KS proofs. 
We finish by a general discussion and conclusion in section~\ref{sec:concl}.


\section{Preliminaries}\label{sec:prel}
Much of the reasoning involving KS sets is usually carried out using a graph representation of KS sets defined below. We thus start by reminding some basic graph-theoretic definitions.

\paragraph*{Graphs.}
Throughout the paper, we will deal with simple undirected finite graphs $G$, i.e., finite graphs without loops, multi-edges or directed edges. We denote $V(G)$ the vertices of $G$ and $E(G)$ the edges of $G$. If two vertices $v_1,v_2$ are connected by an edge, we say that they are adjacent, and write $v_1\sim v_2$.

A subgraph $H$ of $G$ (denoted $H < G$) is a graph formed from a subset of vertices and edges of $G$, i.e., $V(H) \subseteq V(G)$ and $E(H) \subseteq E(G)$. An induced subgraph $K$ of $G$ (denoted $K \lhd G$) is a subgraph that includes all the edges in $G$ whose endpoints belong to the vertex subset $V(K) \subseteq V(G)$, i.e., $E(K) \subseteq E(G)$ with $(v_1, v_2) \in E(K)$ iff $(v_1, v_2) \in E(G)$ for all $v_1, v_2 \in V(K)$.  

A clique in the graph $G$ is a subset of vertices $Q \subset V(G)$ such that every pair of vertices in $Q$ is connected by an edge, i.e., $\forall v_1, v_2 \in Q$ we have $v_1 \sim v_2$. A maximal clique in $G$ is a clique that is not a subset of a larger clique in $G$. 
A maximum clique in $G$ is a clique that is of maximum size in $G$. The clique number $\omega(G)$ of $G$ is the cardinality of a maximum clique in $G$.
 
\paragraph*{Orthogonality graphs.}
The use of graphs in the context of the KS theorem comes from the fact that it is convenient to represent the orthogonality relations in a KS set $\mathcal{S}$ by a graph $G_\mathcal{S}$, known as its orthogonality graph \cite{CSW10, CSW14}. In such a graph, each vector $|v_i\rangle$ in $\mathcal{S}$ is represented by a vertex $v_i$ of $ G_{\mathcal{S}}$ and two vertices $v_1, v_2$ of $G_{\mathcal{S}}$ are connected by an edge if the associated vectors $|v_1 \rangle, |v_2 \rangle$ are orthogonal, i.e. $v_1\sim v_2$ if $\langle v_1 | v_2 \rangle = 0$ (for instance the graph in Fig.~\ref{fig:Clifton} is the orthogonality graph of the set of vectors given by eq.~(\ref{eq:Clif-orth-rep})).

It follows that in an orthogonality graph $G_{\mathcal{S}}$, a clique corresponds to a set of mutually orthogonal vectors in $\mathcal{S}$. If $\mathcal{S}\subset\mathbb{C}^d$ contains a basis set of $d$ orthogonal vectors, then the maximum clique in $G_{\mathcal{S}}$ is of size $\omega(G_{\mathcal{S}}) = d$. 

\paragraph*{Coloring of graphs.}
The problem of $\{0,1\}$-coloring $\mathcal{S}$ thus translates into the problem of coloring the vertices of its orthogonality graph $G_\mathcal{S}$ such that vertices connected by an edge cannot both be assigned the color 1 and maximum cliques have exactly one vertex of color~1. Formally, we say that an arbitrary graph $G$ is $\{0,1\}$-colorable if there exists an assignment $f : V(G) \rightarrow \{0,1\}$ such that 
\begin{flalign}\label{eq:01rulegraph}
\begin{minipage}{0.42\textwidth}
\begin{itemize}
	\item $\sum_{v \in Q} f(v) \leq 1$ for every clique $Q \subset V(G)$;
	\item $\sum_{v \in Q_{\max}} f(v) = 1$ for every maximum clique $Q_{\max} \subset V(G)$.
\end{itemize}
\end{minipage} &&
\end{flalign}
The KS theorem is then equivalent to the statement that there exist for any $d\geq 3$, finite sets of vectors $\mathcal{S} \subset \mathbb{C}^d$ (the KS sets)	such that their orthogonality graph $G_{\mathcal{S}}$ is not $\{0,1\}$-colorable. 
Deciding if a given graph $G$ admits a $\{0,1\}$-coloring is NP-complete \cite{Arends09}.
Note that any graph $G$ that is not $\{0,1\}$-colorable must contain at least two cliques of maximum size $\omega(G)$. Indeed, if a graph $G$ contains a single clique of maximum size it always admits a $\{0,1\}$-coloring consisting in assigning the value 0 to all its vertices, except for one vertex in the maximum clique that is assigned the value 1.

\paragraph*{Orthogonal representations.}
For a given graph $G$, an orthogonal representation $\mathcal{S}$ of $G$ in dimension $d$ is a set of non-zero vectors $\mathcal{S}=\{|v_i \rangle\}$ in $\mathbb{C}^d$ obeying the orthogonality conditions imposed by the edges of the graph, i.e., $v_1 \sim v_2 \Rightarrow \langle v_1|v_2 \rangle=0$ \cite{Lovasz87}. We denote by $d(G)$ the minimum dimension of an orthogonal representation of $G$ and we say that $G$ has dimension $d(G)$. Obviously, $d(G)\geq \omega(G)$. A \emph{faithful} orthogonal representation of $G$ is given by a set of vectors $\mathcal{S}=\{|v_i \rangle\}$ that in addition obey the condition that non-adjacent vertices are assigned non-orthogonal vectors, i.e., $v_1 \sim v_2 \Leftrightarrow  \langle v_1|v_2 \rangle=0$ and that distinct vertices are assigned different vectors, i.e., $v_1 \neq v_2 \Leftrightarrow |v_1 \rangle \neq |v_2 \rangle$.
We denote by $d^*(G)$ the minimum dimension of such a faithful orthogonal representation of $G$ and we say that $G$ has faithful dimension $d^*(G)$.

Given a graph $G$ of dimension $d(G)$, the orthogonality graph $G_\mathcal{S}$ of the minimal orthogonal representation $\mathcal{S}$ of $G$ has faithful dimension $d^*(G_\mathcal{S})=d(G)$. The graph $G_\mathcal{S}$ can be seen as obtained from $G$ by adding edges (between vertices that are non-adjacent in $G$, but corresponding to vectors in $\mathcal{S}$ that are nevertheless orthogonal) and by identifying certain vertices (those that correspond to identical vectors in $\mathcal{S}$). We say that $G_\mathcal{S}$ is the faithful version of $G$.

\paragraph*{KS graphs.}
While the non-$\{0,1\}$-colorability of a set $S$ translates into the non-$\{0,1\}$-colorability of its orthogonality graph $G_\mathcal{S}$, the non-$\{0,1\}$-colorability of an arbitrary graph $G$ translates into the non-$\{0,1\}$-colorability of one of its orthogonal representations only if this representation has the minimal dimension $d(G)=\omega(G)$. Indeed, it is only under this condition that the requirement that $\sum_{v \in Q_{\max}} f(v) = 1$ in the definition of the $\{0,1\}$-coloring of the graph $G$ gives rise to the corresponding requirement that $\sum_{v\in Q_{\max}} f(|v\rangle)  = 1$ for its orthogonal representation (if the dimension $d$ is larger than $\omega(G)=|Q_{\max}|$, the $|Q_{\max}|<d$ mutually orthogonal vectors $\{|v\rangle:v\in Q_{\max}\}$ in $\mathbb{C}^d$ do not form a basis).

If a graph $G$ is not $\{0,1\}$-colorable and has dimension $d(G)=\omega(G)$, it thus follows that its minimal orthogonal representation $\mathcal{S}$ forms a KS set. If in addition $d^*(G)=\omega(G)$, we say that $G$ is a KS graph (this last condition can always be obtained by considering the faithful version of $G$, i.e., the orthogonality graph $G_\mathcal{S}$ of its minimal orthogonal representation $\mathcal{S}$).


The problem of finding KS sets can thus be reduced to the problem of finding KS graphs. But as we have noticed above, deciding if a graph is $\{0,1\}$-colorable is NP-complete. In addition, while finding an orthogonal representation for a given graph can be expressed as finding a solution to a system of polynomial equations, efficient numerical methods for finding such representations are still lacking. Thus, finding KS sets in arbitrary dimensions is a difficult problem towards which a huge amount of effort has been expended \cite{CA96}. In particular, ``records'' of minimal Kochen-Specker systems in different dimensions have been studied \cite{CEG96}, the minimal KS system in dimension four is the $18$-vector system due to Cabello et al. \cite{CEG96, Cab08} while lower bounds on the size of minimal KS systems in other dimensions have also been established.


\section{$01$-gadgets and the Kochen-Specker theorem}\label{sec:gadget}
We now introduce the notion of $01$-gadgets that play a crucial role in constructions of KS sets. 
\begin{dfn}
	A $01$-gadget in dimension $d$ is a $\{0,1\}$-colorable set $\mathcal{S}_\text{gad}\subset\mathbb{C}^d$ of vectors containing two distinguished vectors $|v_1\rangle$ and $|v_2\rangle$ that are non-orthogonal, but for which $f(|v_1\rangle)+f(|v_2\rangle)\leq 1$ in every $\{0,1\}$-coloring $f$ of $\mathcal{S}_\text{gad}$.
\end{dfn}
In other words, while a $01$-gadget $\mathcal{S}_\text{gad}$ admits a $\{0,1\}$-coloring, in any such coloring the two distinguished non-orthogonal vertices cannot both be assigned the value $1$ (as if they were actually orthogonal).
We can give an equivalent, alternative definition of a gadget as a graph.
\begin{dfn}
	A $01$-gadget in dimension $d$ is a $\{0,1\}$-colorable graph $G_\text{gad}$ with faithful dimension $d^*(G_\text{gad})=\omega(G_\text{gad})=d$ and with two distinguished non-adjacent vertices $v_1 \nsim v_2$ such that $f(v_1)+f(v_2)\leq 1$ in every $\{0,1\}$-coloring $f$ of $G_\text{gad}$.
\end{dfn}
In the following when we refer to a $01$-gadget, we freely alternate between the equivalent set or graph definitions.

An example of a $01$-gadget in dimension 3 is given by the following set of 8 vectors in $\mathbb{C}^3$:
\begin{eqnarray}
	\label{eq:Clif-orth-rep}
&&	| u_1 \rangle = \frac{1}{\sqrt{3}}(-1,1,1), \; \; |u_2 \rangle = \frac{1}{\sqrt{2}}(1,1,0), \nonumber \\ 
&& |u_3 \rangle = \frac{1}{\sqrt{2}} (0,1,-1), |u_4 \rangle = (0,0,1), \nonumber \\
&&	|u_5 \rangle = (1,0,0), \; \; |u_6 \rangle = \frac{1}{\sqrt{2}}(1,-1,0), \nonumber \\
&& |u_7 \rangle = \frac{1}{\sqrt{2}}(0,1,1), \; \; |u_8 \rangle = \frac{1}{\sqrt{3}}(1,1,1), 
	\end{eqnarray}
where the two distinguished vectors are $|v_1\rangle=|u_1\rangle$ and $|v_2\rangle=|u_8\rangle$. Its  orthogonality graph is represented in Fig.~\ref{fig:Clifton}. It is easily seen from this graph representation that the vertices $u_1$ and $u_8$ cannot both be assigned the value 1, as this then necessarily leads to the adjacent vertices $u_4$ and $u_5$ to be both assigned the value 1, in contradiction with the $\{0,1\}$-coloring rules. This graph was identified by Clifton, following work by Stairs \cite{Clifton93, Stairs}, and used by him to construct statistical proofs of the Kochen-Specker theorem. We will refer to it as the Clifton gadget $G_{\text{Clif}}$. The Clifton gadget and similar gadgets were termed ``definite prediction sets" in \cite{CA96}. 
\begin{figure}[t]
	\centerline{\includegraphics[scale=0.33]{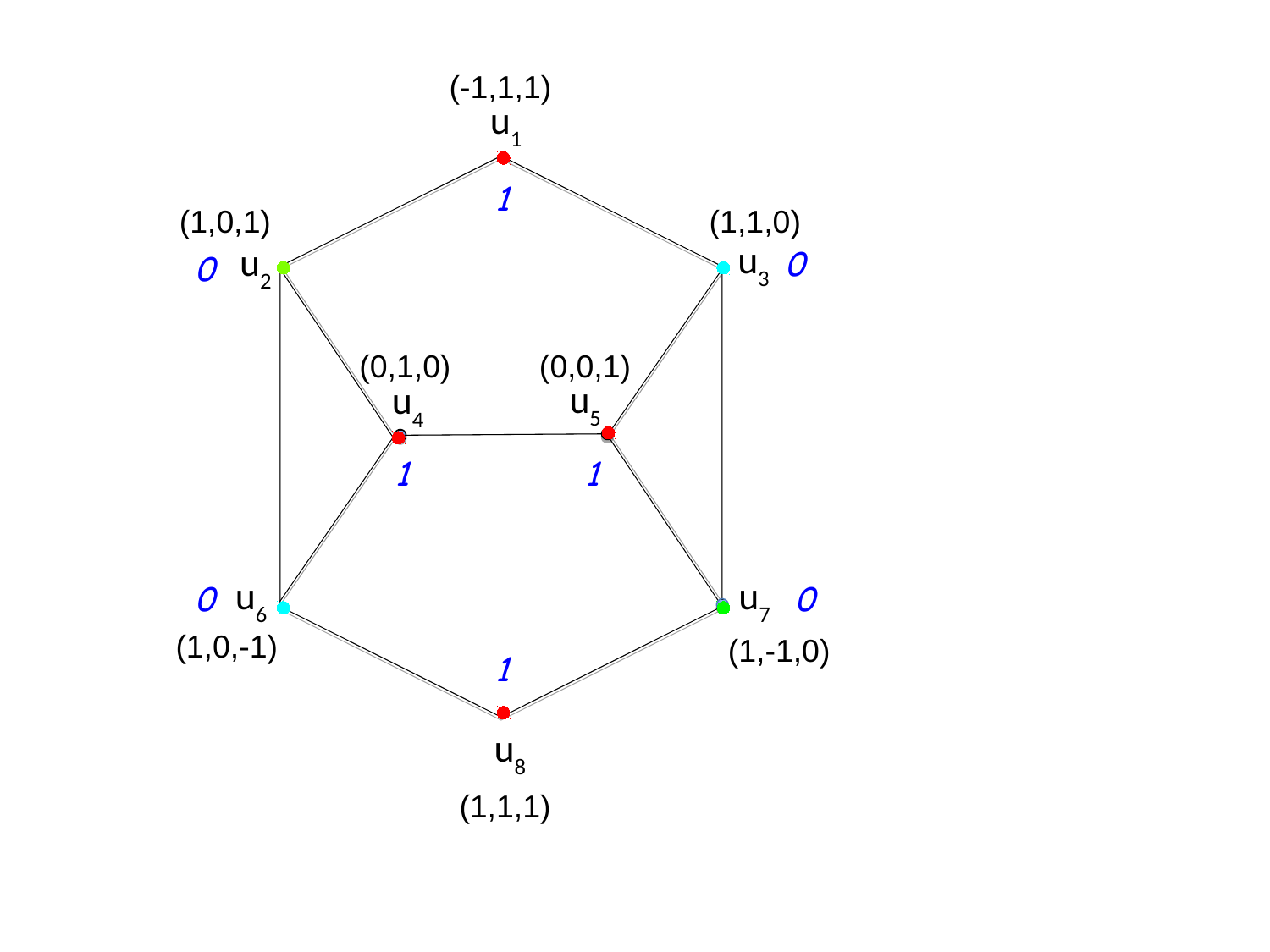}}
	\caption{The $8$-vertex ``Clifton" graph that was used by Kochen and Specker in their construction of the $117$ vector KS set. The two distinguished vertices are $u_1$ and $u_8$.}
	\label{fig:Clifton}
\end{figure}

We identify the role played by 01-gadgets in the construction of Kochen-Specker sets via the following theorem. 
\begin{theorem}
	\label{prop:KS-gadg}
	For any Kochen-Specker graph $G_\text{KS}$, there exists a subgraph $G_{\text{gad}}<G_{\text{KS}}$ with $\omega(G_\text{gad})=\omega(G_\text{KS})$ that is a $01$-gadget. Moreover, given a $01$-gadget $G_{\text{gad}}$, one can construct a KS graph $G_\text{KS}$ with $\omega(G_\text{KS})=	\omega(G_{\text{gad}})$. 
\end{theorem} 
The demonstration of our theorem is constructive, it allows to build a $01$-gadget from a KS graph and conversely.
The $01$-gadget in the original $117$-vector proof by Kochen-Specker is the Clifton graph in Fig.~\ref{fig:Clifton}.
A $16$-vertex $01$-gadget in dimension 4 that is an induced subgraph of the $18$-vertex KS graph introduced in \cite{CEG96} is represented in Fig. \ref{fig:cab-KS-gadget}.
\begin{figure}[t]
	\centerline{\includegraphics[scale=0.39]{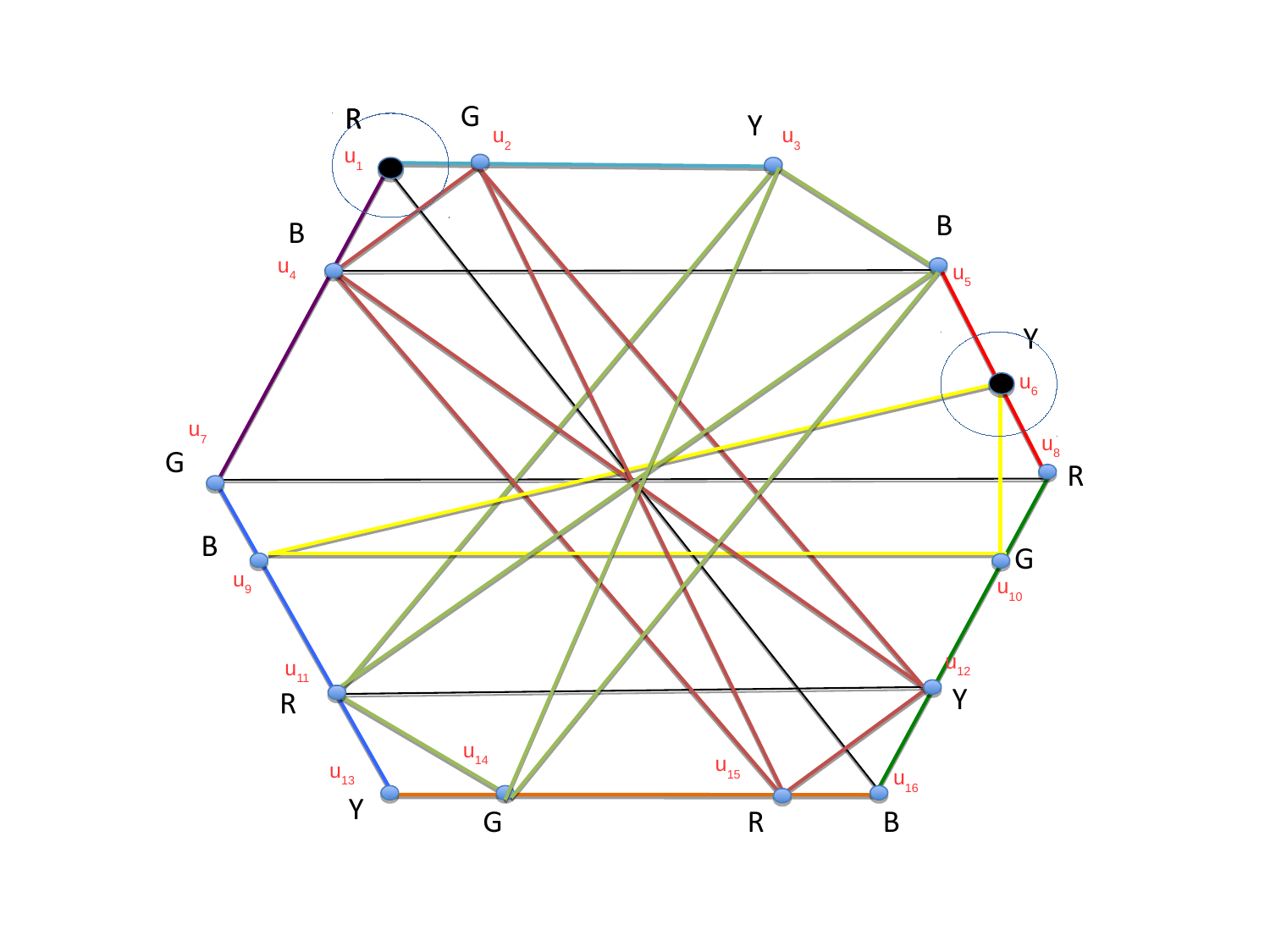}}
	\caption{A $16$ vertex coloring gadget (also a $101$-gadget) that is a subgraph of the $18$ vertex Kochen-Specker graph in dimension $d=4$ found by Cabello et al. \cite{CEG96}. The $9$ edge colors denote $9$ cliques in the graph, with the maximum clique being of size $\omega(G) = 4$. The distinguished vertices $u_1, u_6$ are denoted by black circles.}
	\label{fig:cab-KS-gadget}
\end{figure} 

\begin{proof}
We start by showing the first part of the Theorem: that one can construct a $01$-gadget $G_{\text{gad}}$ from any KS graph $G_\text{KS}$.  Given $G_\text{KS}$, which by definition is not $\{0,1\}$-colorable, we first construct, by deleting vertices one at a time, an induced subgraph $G_{\text{crit}}$ that is vertex-critical. By vertex-critical, we mean that $(i)$ $G_{\text{crit}}$ is not $\{0,1\}$-colorable, but $(ii)$ any subgraph obtained from it by deleting a supplementary vertex does admit a $\{0,1\}$-coloring. Observe that in the process of constructing $G_{\text{crit}}$ we are able to preserve the maximum clique size, i.e., $\omega(G_{\text{crit}}) = \omega(G_{\text{KS}})$. This is because we are able to delete vertices from all but two maximum cliques, simply because at least two maximum cliques must exist in a graph that is not $\{0,1\}$-colorable. Observe also that $G_{\text{crit}}$ is itself a KS graph, since the faithful orthogonal representation of $G_\text{KS}$ in dimension $d=\omega(G_{\text{	KS}})$ provides an orthogonal representation of $G_\text{crit}$ in the same dimension. 

We consider three cases: $(i)$ there exists a vertex in $G_{\text{crit}}$ that belongs to a single maximum clique, $(ii)$ all vertices in $G_{\text{crit}}$ belong to at least two maximum cliques, and there exists a vertex that belong to exactly two maximum cliques; $(iii)$, all vertices in $G_{\text{crit}}$ belong to at least three maximum cliques. In the first two cases, which happen to be the case encountered in all known KS graphs, we will be able to prove that the $01$-gadget appears as an induced subgraph while in the third case, the $01$-gadget appears as a subgraph that may not necessarily be induced. 

In case $(i)$, let $v$ be one of the vertices having the property that it belongs to a single maximum clique. We denote this clique $Q_1 \subset G_{\mathcal{S}}^{\text{crit}}$. Deleting $v$ leads to a graph $G_{\text{crit}} \setminus v$ that is $\{0,1\}$-colorable by definition. However, observe that in any coloring $f$ of $G_{\text{crit}} \setminus v$, all the vertices in $Q_1 \setminus v$ are assigned the value $0$ by $f$. This is because, if one of these vertices were assigned value $1$, then one could obtain a valid coloring of $G_{\text{crit}}$ from $f$ by defining $f(v) = 0$. Choose a vertex $v_1\in Q_1 \setminus v$ and any other non-adjacent vertex $v_2\in G_{\text{crit}} \setminus v$. Then $G_{\text{crit}} \setminus v$ is the required $01$-gadget with $v_1, v_2$ playing the role of the distinguished vertices.  

In case $(ii)$, let $v$ be one of the vertices having the property that it belongs to exactly two maximum cliques, which we denote $Q_1, Q_2 \subset G_{\text{crit}}$. Again, deleting $v$ leads to a graph $G_{\text{crit}} \setminus v$ that is $\{0,1\}$-colorable. However, in any coloring $f$ of $G_{\text{crit}} \setminus v$, it cannot be that a value $f(v_1)=1$ and a value $f(v_2)=1$ are simultaneously assigned to a vertex $v_1 \in Q_1 \setminus v$ and a vertex $v_2 \in Q_2 \setminus v$. This is again because if there was such a coloring $f$, then one could obtain a valid coloring for $G_{\text{crit}}$ by defining $f(v) = 0$, in contradiction with the criticality of $G_{\text{crit}}$. Choose $v_1 \in Q_1 \setminus v$ and $v_2 \in Q_2 \setminus v$ such that $v_1$ and $v_2$ are not adjacent. Two such vertices must exist. Indeed, if all vertices $Q_1 \setminus v$ where adjacent to all vertices of $Q_2 \setminus v$, then the maximum clique size would be strictly greater than $\omega(G_\text{crit})$. Therefore, we have that $G_{\text{crit}} \setminus v$ is the required $01$-gadget with $v_1, v_2$ the distinguished vertices. 

Finally, we consider the case $(iii)$ where each vertex in $G_{\text{crit}}$ belongs to at least three maximum cliques. In this case, we cannot proceed as above where we remove a certain vertex $v$ and pick vertices from two maximal cliques containing $v$, because we can no longer guarantee that these two vertices cannot simultaneously be assigned the value 1 (we can only guarantee that a certain $t$-uple of vertices, each one picked from the $t$ maximum cliques to which $v$ belongs, cannot all simultaneously be assigned the value 1, which may be thought of as a generalization of the $01$-gadget property to $t$ distinguished vertices in place of two). Instead, we proceed as follows. 
We start by deleting edges of $G_{\text{crit}}$ one at a time, to construct a new graph $G'_{\text{crit}}$ that is edge-critical. By edge-critical, we mean, similarly to the construction above, that $G'_{\text{crit}}$ is not $\{0,1\}$-colorable, but any graph obtained from it by deleting a supplementary edge (and thus also by deleting a supplementary vertex) does admit a $\{0,1\}$-coloring. As above, we are able to preserve the maximum clique size in the process, i.e., $\omega(G'_{\text{crit}}) = \omega(G_{\text{crit}})=\omega(G_\text{KS})$, and $G'_{\text{crit}}$ is still a non-$\{0,1\}$-colorable KS graph.

Case (iii a): If the resulting graph $G'_{\text{crit}}$ is as in the cases $(i)$ and $(ii)$ above, we proceed as before to construct a $01$-gadget from a graph $G'_{\text{crit}}\setminus v$, with the caveat that choosing two non-adjacent vertices $v_1$ and $v_2$ in $G'_{\text{crit}}$ does not necessarily guarantee that they correspond to non-orthogonal vectors in the natural representation induced by the one of $G_\text{KS}$. This is because we have been removing edges from $G_\text{crit}$ to construct $G'_\text{crit}$. However, we can always choose two vertices $v_1$ and $v_2$ that were non-adjacent in the original graph $G_\text{KS}$ and that thus correspond to non-orthogonal vectors. Again, this is because otherwise the maximum clique size of $G_\text{KS}$ would be greater than $\omega(G_\text{KS})$. Now, in any $\{0,1\}$-coloring of $G'_{\text{crit}}\setminus v$, we cannot have both $f(v_1) = 1$ and $f(v_2) = 1$, so that $G'_{\text{crit}}\setminus v$ forms a subgraph of $G$ that is a $01$-gadget. But notice that the $\{0,1\}$-colorings of $G_{\text{crit}}\setminus v$ are a subset of the $\{0,1\}$-colorings of $G'_{\text{crit}}\setminus v$. So that we cannot have both $f(v_1) = 1$ and $f(v_2) = 1$ in any $\{0,1\}$-coloring of $G_{\text{crit}}\setminus v$ as well. So that the $01$-gadget is given by $G_{\text{crit}}\setminus v$ in this case with $v_1, v_2$ the distinguished vertices.




Case(iii b): If the resulting graph $G'_{\text{crit}}$ is not as in the cases $(i)$ and $(ii)$ above, we proceed as follows. 
Let $v$ be an arbitrary vertex of $G'_{\text{crit}}$. By assumption, this vertex belong to at least two maximun cliques $Q_1, Q_2$ (and actually even at least a third one). Delete all the edges $(v,v')$ from $Q_1$ where $v' \in Q_1$ to form $G'_{\text{crit}} \setminus E_v(Q_1)$ (where $E_v(Q_1)$ denotes the edges incident on $v$ in $Q_1$) which is $\{0,1\}$-colorable by definition. In any such coloring $f$, either $f(v)=0$ or $f(v)=1$. In the first case, we must necessarily have that $f(v')=0$ for all $v' \in Q_1\setminus v$, since otherwise the coloring $f$ would also define a valid coloring for $G'_{\text{crit}}$. In the second case, we have $f(v'')=0$ for all $v''\in Q_2\setminus v$ by definition of a coloring. We thus conclude that it cannot be simultaneously the case that $f(v')=f(v'')=1$ for $v'\in Q_1\setminus v$ and $v''\in Q_2\setminus v$. Choose $v_1 \in Q_1$ and $v_2\in Q_2$ non-adjacent in $G_\text{KS}$, which is always possible by the same argument as given before. The faithful version of the graph $G'_{\text{crit}} \setminus E_v(Q_1)$ forms the required $01$-gadget with $v_1, v_2$ being the distinguished vertices. Indeed, by the preceding argument, one can restore edges from $G_{\text{crit}}$ to the graph $G'_{\text{crit}} \setminus E_v(Q_1)$ to obtain the $01$-gadget so long as the graph is $\{0,1\}$-colorable, an instance of this is the graph $G'_{\text{crit}} \setminus (v, v_1)$.

We now proceed to prove the second part of the statement. Starting from a gadget graph we give a construction of a KS graph. The construction generalizes the original Kochen-Specker construction of \cite{KS} to arbitrary dimensions and arbitrary repeating gadget units. Given $G_{\text{gad}}$, we know that there exists a faithful orthogonal representation $\{| v_i \rangle\}_{i=1}^{n}$ in a Hilbert space of dimension $d = \omega(G_{\text{gad}})$ with $n = |V(G_{\text{gad}})|$. Let $v_1, v_2$ denote the distinguished vertices, and let $|v_2^{\perp} \rangle$ denote a vector orthogonal to $|v_2 \rangle$ that lies in the plane $\text{span}(|v_1 \rangle, |v_2 \rangle)$, spanned by the vectors $|v_1 \rangle$ and $|v_2 \rangle$, with $\theta = \arccos{|\langle v_1| v_2^{\perp} \rangle|}>0$ by definition of a $01$-gadget. We consider the following cases: $(i)$ $\frac{\pi}{2 \theta}$ is rational and can be written as $\frac{p}{q}$ with $q$ an odd integer, $(ii)$ $\frac{\pi}{2 \theta}$ is rational and is given by $\frac{p}{q}$ with $q$ an even integer, or alternatively, $\frac{\pi}{2 \theta}$ is irrational.  
 
\emph{Case $(i)$: $\frac{\pi}{2 \theta}$ is rational and is given by $\frac{p}{q}$ with $q$ an odd integer.}
	 Recall that $|v_2^{\perp} \rangle$ is orthogonal to $|v_2 \rangle$ in the plane $\text{span}(|v_1 \rangle, |v_2 \rangle)$. In the subspace orthogonal to $\text{span}(|v_1 \rangle, |v_2 \rangle)$, choose a basis consisting of $d-2$ mutually orthogonal vectors $|w_1 \rangle, \dots, |w_{d-2} \rangle$. Denoting $G'_{\text{gad}}$ as the orthogonality graph of the entire set of these vectors $\{| v_i \rangle\}_{i=1}^{n} \bigcup \{|v_2^{\perp} \rangle, |w_1 \rangle, \dots, |w_{d-2} \rangle\}$, we obtain a gadget graph that can be used as a building block in a Kochen-Specker type construction. In particular, the crucial property of $G'_{\text{gad}}$ is that in any $\{0,1\}$-coloring $f$, $f(v_1) = 1 \Rightarrow f(v_2^{\perp}) = 1$. This can be seen as follows: $f(v_1) = 1$ implies, by the $\{0,1\}$-coloring rules, that $f(w_i) = 0$ for all $i \in [d-2]$. Moreover, by the gadget property, we have $f(v_2) = 0$, and this imposes $f(v_2^{\perp}) = 1$ to satisfy the requirement that exactly one of the vertices in the maximum clique $(v_2,v_2^{\perp},w_1,\dots, w_{d-2})$ is assigned value $1$. 
	
	As in the original KS construction of \cite{KS}, we construct a chain of $p+1$ copies $G'^{(i)}_{\text{gad}}$ $(i=0,1,\ldots,p\}$ of $G'_{\text{gad}}$ so that $p \theta = q \frac{\pi}{2}$ is an odd integral multiple of $\frac{\pi}{2}$. These copies are obtained from the realization of $G'_{\text{gad}}$ by successive applications of a unitary $\mathcal{U}$, i.e., $|v_j^{(i)} \rangle = \mathcal{U}^{i} | v_j \rangle$ for $i=0,1,\ldots,p$ and $j=1,\ldots,n$ and similarly for the other vectors in $G'_{\text{gad}}$. This unitary operator $\mathcal{U}$ is defined as
	\begin{eqnarray}
	\mathcal{U} = |v_2^{\perp} \rangle \langle v_1 | - | v_2 \rangle \langle v_1^{\perp}| + \textbf{1}_W, 
	\end{eqnarray}
	where $| v_1^{\perp} \rangle$ denotes the vector orthogonal to $|v_1 \rangle$ in the plane $\text{span}(|v_1 \rangle, |v_2 \rangle)$ and where $\textbf{1}_W$ denotes the identity on the subspace orthogonal to $\text{span}(|v_1 \rangle, |v_2 \rangle)$. Writing $|v_2^{\perp} \rangle = \alpha | v_1 \rangle + \beta | v_1^{\perp} \rangle$ for some $\alpha, \beta \in \mathbb{C}$, we see that applying once $\mathcal{U}$ to the faithful realization of $G'_{\text{gad}}$ gives
	\begin{eqnarray}
	\mathcal{U} | v_1 \rangle = | v_2^{\perp} \rangle, \nonumber \\
	\mathcal{U} | v_2^{\perp} \rangle = \alpha | v_2^{\perp} \rangle - \beta | v_2 \rangle.
	\end{eqnarray}
	We have evidently $| \langle v_2^{\perp} | \mathcal{U} | v_2^{\perp} \rangle| = |\langle v_1 | v_2^{\perp} \rangle|$ and that 
	\begin{eqnarray}
	\arccos|\langle v_1 | \mathcal{U}| v_2^{\perp} \rangle| = 2 \arccos|\langle v_1 | v_2^{\perp} \rangle| = 2 \theta.
	\end{eqnarray}
We thus have that under successive applications of $\mathcal{U}$, $|v_1^{(0)}\rangle\rightarrow |v_1^{(1)}\rangle=|v_2^{\perp,(0)}\rangle$, $|v_2^{\perp,(0)}\rangle\rightarrow |v_2^{\perp,(1)}\rangle$, $|v_1^{(1)}\rangle\rightarrow |v_1^{(2)}\rangle=|v_2^{\perp,(1)}\rangle$, $|v_2^{\perp,(1)}\rangle\rightarrow |v_2^{\perp,(2)}\rangle$, and so on, with $|v_1^{(p)} \rangle \perp |v_1^{(0)} \rangle$. Furthermore, in any $\{0,1\}$-coloring $f$ of the graph union $\bigcup_{i} G'^{(i)}_{\text{gad}}$, $f(v_1^{(0)})=1 \Rightarrow f(v_1^{(p)}) = 1$. A similar construction of $d-1$ copies of $\bigcup_{i} G'^{(i)}_{\text{gad}}$ gives rise to a graph with a clique formed by the vertices $v_1^{(0)}, v_1^{(p)}$ and the $d-2$ vectors that complete the basis. The resulting graph is a Kochen-Specker graph since in any $\{0,1\}$-coloring, if any of the vertices in this maximal clique is assigned value $1$ then so are all of them, giving rise to a contradiction. We thus obtain a finite system of vectors given by the union of the vector sets in each of the graphs, that gives rise to a proof of the Kochen-Specker theorem in dimension $\omega(G_{\text{gad}})$.

\emph{Case $(ii)$: $\frac{\pi}{2 \theta}$ is rational and is given by $\frac{p}{q}$ with $q$ an even integer, or alternatively, $\frac{\pi}{2 \theta}$ is irrational.}
	
	In this case, we construct from $G_{\text{gad}}$ a larger gadget $\tilde{G}_{\text{gad}}$ with the property that the angle $\tilde{\theta}$ between the distinguished vectors obeys $\frac{\pi}{2 \tilde{\theta}} = \frac{\tilde{p}}{\tilde{q}} \in \mathbb{Q}$, with $\tilde{q}$ an odd integer. 
	As in the previous case, we let $|v_2^{\perp} \rangle$ be the vector orthogonal to $|v_2 \rangle$ in the plane $\text{span}(|v_1 \rangle, |v_2 \rangle)$, and $| v_1^{\perp} \rangle$ be the vector orthogonal to $|v_1 \rangle$ in this plane, so that $|v_2^{\perp} \rangle = \alpha | v_1 \rangle + \beta | v_1^{\perp} \rangle$,
	for some $\alpha, \beta \in \mathbb{C}$. We also consider a basis $\{|w_1 \rangle, \dots, |w_{d-2} \rangle \}$ for the subspace orthogonal to $\text{span}(|v_1 \rangle, |v_2 \rangle)$ and denote $G'_{\text{gad}}$ as the orthogonality graph of the set of vectors $\{|v_i \rangle\}_{i=1}^n \bigcup \{|v_2^{\perp} \rangle, |w_1 \rangle, \dots, |w_{d-2}\rangle\}$.
	
	Let $\mathcal{U}$ denote a unitary operator transforming $|v_1 \rangle$ to $|v_2^{\perp} \rangle$, i.e., $\mathcal{U}$ is of the form 
	\begin{eqnarray}
	\mathcal{U} &=& |v_2^{\perp} \rangle \langle v_1| - | v'_2 \rangle \langle v_1^{\perp}| +  \nonumber \\
&&+  |w'_1 \rangle \langle w_1 |+\dots+ |w'_{d-2} \rangle \langle w_{d-2} | 
	\end{eqnarray}
	with $|v'_2\rangle$, $|w'_1\rangle,\ldots,|w'_{d-2}\rangle$ orthogonal to $|v_2^{\perp} \rangle$ and orthogonal to each other.	
	Applying $\mathcal{U}$ to the orthogonal representation of the gadget gives that 
	\begin{eqnarray}
	\mathcal{U} | v_1 \rangle &=& | v_2^{\perp} \rangle, \nonumber \\
	\mathcal{U} | v_2^{\perp} \rangle &=& \alpha |v_2^{\perp} \rangle - \beta |v'_2 \rangle
	\end{eqnarray} 
	Let $\tilde{\theta} = \arccos| \langle v_1 | \mathcal{U} | v_2^{\perp} \rangle |$.
	We choose $|v'_2\rangle$ and thereby $\mathcal{U}$ such that $\frac{\pi}{2 \tilde{\theta}} = \frac{\tilde{p}}{\tilde{q}} \in \mathbb{Q}$ with $\tilde{q}$ an odd integer.  Now construct $G'_{\text{gad}}$ as the orthogonality graph of the set of vectors 
	\begin{eqnarray}
	\{|v_i \rangle\}_{i=1}^{n} \bigcup \{|v_2^{\perp} \rangle, |w_1 \rangle, \dots, |w_{d-2}\rangle\} \bigcup \nonumber \\ \{\mathcal{U}|v_i \rangle\}_{i=2}^{n} \bigcup \{\mathcal{U} |v_2^{\perp} \rangle, |w'_1 \rangle, \dots, |w'_{d-2} \rangle\}.
	\end{eqnarray}  
	We have thus concatenated two gadgets to form the new gadget $G'_{\text{gad}}$ with the property that if $f(|v_1 \rangle) = 1$ then also $f(|v_2^{\perp} \rangle) = 1$ and consequently also $f(\mathcal{U} | v_2^{\perp} \rangle) = 1$. We are now in the same position as in the previous case i.e., we may construct a chain of $\tilde{p}+1$ copies $G'^{(i)}_{\text{gad}}$ of $G'_{\text{gad}}$ and follow the steps as in the previous case to construct the entire KS set in dimension $\omega(G_{\text{gad}})$.

In both cases, we thus obtain a construction of a Kochen-Specker set in dimension $\omega(G_{\text{gad}})$, completing the proof. 
\end{proof}

We remark that the above Theorem does not guarantee that the $01$-gadgets appear as \textit{induced} subgraphs in KS graphs; this is the case only when every vertex in the $\{0,1\}$-edge-critical subgraph of the KS graph does not belong to three or more maximum cliques (cases $(i), (ii)$ and $(iii a)$ in the proof). As such, in the case $(iii b)$ where every vertex in the $\{0,1\}$-edge-critical subgraph of the KS graph belongs to at least three maximum cliques, the subgraphs may not correspond to vector subsets of the original KS vector set. We leave it as an interesting open question whether $01$-gadgets always appear as vector subsets of the KS vector sets in this case as well. We also note that constructions similar to that given in the proof of the second part of Theorem \ref{prop:KS-gadg} have appeared in \cite{BBCP09}.

\section{Other 01-gadgets and KS sets constructions}\label{sec:constr}
In this section, we make some interesting observations about $01$-gadgets and provide new constructions of $01$-gadgets that will be used in the next sections. 

\begin{lemma}
	\label{lem:min-gad}
	For any $d\geq 3$, there exists a $01$-gadget in dimension $d$ consisting of $5+d$ vertices.    
\end{lemma}
\begin{proof}
For $d=3$, a 8-vertex $01$-gadget is simply given by the Clifton gadget $G_{\text{Clif}}$. In higher dimensions, a new $01$-gadget $G'_{\text{Clif}}$ can be obtained by adding $d-3$ vertices to $G_{\text{Clif}}$ with edges joining the additional vertices to each other and to each of the $8$ vertices in $G_{\text{Clif}}$. Clearly, a faithful representation of $G'_{\text{Clif}}$ can be obtained by supplementing the $3$-dimensional representation of $G_{\text{Clif}}$ with $d-3$ mutually orthogonal vectors in the complementary subspace. The construction preserves the property that a $\{0,1\}$-coloring of $G'_{\text{Clif}}$ exists and that the two distinguished vertices $v_1,v_2$ of $G_{\text{Clif}}$, now viewed 	as vertices of $G'_{\text{Clif}}$, cannot both be assigned the value $1$ in any $\{0,1\}$ coloring. 
\end{proof}

The $8$-vertex Clifton gadget $G_{\text{Clif}}$ was shown to be the minimal $01$-gadget in dimension 3 \cite{Arends09}. This result was obtained by an exhaustive search over all non-isomorphic square-free graphs of up to $7$ vertices. 
It is an open question to prove if the simple construction in Lemma \ref{lem:min-gad} gives the minimal $01$-gadgets in dimension $d>3$ or whether even smaller gadgets exist in these higher dimensions. 

In the Clifton gadget $G_{\text{Clif}}$ the overlap between the two distinguised vertices is $|\langle v_1|v_2\rangle|=1/3$. The following Lemma shows that one can reduce this overlap at the expense of increasing the dimension by one. 
\begin{lemma}
	\label{lem:gad-realize}
	Let $G$ be a $01$-gadget in dimension $d$ with distinguished vectors $|u_1\rangle, |u_2\rangle$. Then there exists a $01$-gadget $G'$ in dimension $d+1$ with distinguished vertices $|v_1\rangle, |v_2\rangle$ for any choice of the overlap $0<|\langle v_1|v_2\rangle|\leq |\langle u_1|u_2\rangle|$.
\end{lemma}
\begin{proof}
Let $\{|u_i\rangle\}_{i=1}^n\subset \mathbb{C}^d$ be the set of $n$ vectors forming the gadget $G$. We define $G'$ as the set of $n+1$ vectors $\{|v_i\rangle\}_{i=0}^n$ in $\mathbb{C}^{d+1}$ defined as follows. For given $|u_i \rangle \in \mathbb{C}^{d}$, let $|\tilde{u}_i \rangle \in \mathbb{C}^{d+1}$ be the vector obtained by padding a $0$ to the end of $|u_i \rangle$. Define the vectors $|v_i \rangle$ as
	\[
	|v_i \rangle := \left\{\begin{array}{lr}
	(0,\dots,0, 1)^T, & \text{for } i = 0\\
	\mathcal{N}\left(|\tilde{u}_1\rangle + x (0,\dots,0, 1)^T\right), & \text{for } i=1\\
	|\tilde{u}_i \rangle & \text{for } i = 2, \dots, n
	\end{array}\right.
	\] 
	with a free parameter $x \in \mathbb{R}$ and corresponding normalization factor $\mathcal{N}$. 
Now, notice that the orthogonality relations between the set of vectors $|v_1 \rangle, \dots, |v_{n} \rangle$ is the same as the orthogonality relations between the set of vectors $|u_1 \rangle, \dots, |u_{n} \rangle$. The only additional orthogonality relations in $G'$ involve $|v_0\rangle$, which is orthogonal to all other vectors but $|v_1\rangle$. By this property, it follows that if $f(|v_0\rangle)=0$ in a coloring of $G'$, then the coloring of the remaining vectors $|v_1 \rangle, \dots, |v_{n} \rangle$ is constrained exactly as for $|u_1 \rangle, \dots, |u_{n|} \rangle$ in $G$. In particular, we cannot have simultaneously $f(|v_1\rangle)=f(|v_2\rangle)=1$. Now simply observe that if $f(|v_2\rangle)=1$, we must have necessarily have $f(|v_0\rangle)=0$ since $|v_0 \rangle \perp |v_2 \rangle$ and thus $|v_1\rangle$ cannot also satisfy $f(|v_1\rangle) = 1$. In other words, $G'$ is a $01$-gadget with $|v_1\rangle, |v_2\rangle$ playing the role of the distinguished vertices. Finally, we see that by varying the free parameter $x \in \mathbb{R}$, we get  any overlap $0< |\langle v_1|v_2\rangle|\leq |\langle u_1|u_2\rangle|$ between the distinguished vertices. 
\end{proof}

We now show the following.  

\begin{theorem}
	\label{prop:fin-gadg-const}
	Let $|v_1 \rangle$ and $|v_2 \rangle$ be any two distinct non-orthogonal vectors in $\mathbb{C}^d$ with $d \geq 3$. Then there exists a 01-gadget in dimension $d$ with $|v_1 \rangle$ and $|v_2 \rangle$ being the two distinguished vertices.
\end{theorem}   
While the existence of such a construction can be anticipated from the Kochen-Specker construction from Theorem~\ref{prop:KS-gadg}, we give a construction with much fewer vectors based on the $43$-vertex graph of Fig. \ref{fig:gadg-id-vec}.
\begin{proof}
The construction is based on the $43$-vertex graph $G$ of Fig. \ref{fig:gadg-id-vec}. We first show the construction for $\mathbb{C}^3$, and then straightforwardly extend it to $\mathbb{C}^d$ for $d > 3$. Suppose thus that we are given $|v_1 \rangle, |v_2 \rangle \in \mathbb{C}^3$. We consider two cases: (i) $0 < | \langle v_1 | v_2 \rangle| \leq \frac{1}{\sqrt{2}}$ and (ii)  $\frac{1}{\sqrt{2}} < | \langle v_1 | v_2 \rangle| \leq 1$. 
	
Case (i): $0 < | \langle v_1 | v_2 \rangle| \leq \frac{1}{\sqrt{2}}$.
	Suppose without loss of generality that $|v_1 \rangle = (1,0,0)^T$ and $|v_{2} \rangle = \frac{1}{\sqrt{1+x^2}}(x,1,0)^T$ with $0 < x \leq 1$.
	In this case, the induced subgraph $G_{\text{ind}}$ of $G$  consisting of the vertex set $V(G_{\text{ind}}) = \{1,\dots,22\}$ and $E(G_{\text{ind}}) = \{(u_i, u_j): 1 \leq i,j \leq 22, (u_i,u_j) \in E(G)\}$ will suffice to construct the gadget with $u_1$ and $u_{22}$ the two distinguished vertices, corresponding to $|v_1 \rangle$ and $|v_2 \rangle$. First, it is easily verified from the graph that in any $\{0,1\}$-coloring $f$,  $f(u_1)$ and $f(u_{22})$ cannot both be assigned the value 1. It thus only remains to provide an orthogonal representation of the graph $G_{\text{ind}}$. Such a representation is given by the following set of (non-normalized) vectors:
	\begin{eqnarray}
	&&|u_1 \rangle = (1,0,0)^T; \; \; |u_2 \rangle = (0,1,-1)^T; \; \; |u_3 \rangle = (0,1,0)^T; \nonumber \\
	&&|u_4 \rangle = (0,y,1)^T; \; \; |u_5 \rangle = (2 x,1,1)^T; \; \; |u_6 \rangle = (-1,0,2 x)^T; \nonumber \\
	&&|u_7 \rangle = (-2 x,0,-1)^T; \; \; |u_8 \rangle = (x,1,-2 x^2)^T; \nonumber \\
	&& |u_9 \rangle = (2x^3, 2 x^2,1+x^2)^T; \nonumber \\
	&& |u_{10} \rangle = (-(1+x^2),0,2x^3)^T; \nonumber \\
	&&|u_{11} \rangle = (2 x^3, 0, 1+x^2)^T; \nonumber\\
	&& |u_{12} \rangle = (x(1+x^2), 1+x^2,-2x^4)^T; \nonumber \\
	&&|u_{13} \rangle = (2 x^5, 2x^4, (1+x^2)^2)^T; \nonumber \\
	&& |u_{14} \rangle = (-(1+x^2)^2, 0, 2x^5)^T; \nonumber \\
	&&|u_{15} \rangle = (2x^5, 0,(1+x^2)^2)^T; \nonumber \\
	&& |u_{16} \rangle = (x(1+x^2)^2, (1+x^2)^2, -2x^6)^T; \nonumber \\
	&&|u_{17} \rangle = (2x^7, 2x^6, (1+x^2)^3)^T; \nonumber \\
	&& |u_{18} \rangle = (-x(1+ y^2), -1, y)^T; \nonumber \\
	&&|u_{19} \rangle = (1,-x,-x)^T; \; \; |u_{20} \rangle = (1,-x,0)^T; \nonumber \\
	&&|u_{21} \rangle = (1,-x,x y)^T; \; \; |u_{22} \rangle = (x,1,0)^T;
	\end{eqnarray}
	with 
	\begin{eqnarray}
	y = \frac{(1+x^2)^3 + \sqrt{(1+x^2)^6 - 16 x^{14} (1+x^2)}}{4 x^8}\,.
	\end{eqnarray}
	
It is easily verified that this set of vectors satisfy all the orthogonality relations encoded by the induced subgraph $G_{\text{ind}}$ we are considering.
	
Case (ii):  $\frac{1}{\sqrt{2}} < | \langle v_1 | v_2 \rangle| \leq 1$.
	Suppose without loss of generality that $|v_1 \rangle = (1,0,0)^T$ and $|v_{2} \rangle = (1+x,1-x,0)^T/\sqrt{2+2x^2}$ with $0 < x \leq 1$. In this case, we consider the entire $43$-vertex graph $G$ from Fig. \ref{fig:gadg-id-vec}, with  $u_1$ and $u_{42}$ the two distinguished vertices, corresponding to $|v_1 \rangle$ and $|v_2 \rangle$. Again, it is easily seen that in any $\{0,1\}$-coloring $f$,  $f(u_1)$ and $f(u_{42})$ cannot both be assigned the value 1. It thus only remains to provide an orthogonal representation of the graph $G$. 
	
The graph $G$ can be seen as being composed from $(i)$ the induced subgraph $G_{\text{ind}}$ with vertices $u_1,\ldots,u_{22}$ considered above, $(ii)$ an isomorphic subgraph $G'_{\text{ind}}$ with vertices $u'_1=u_{20},u'_2=u_{23},\ldots,u'_{22}=u_{42}$, $(iii)$ the vertex $u_{43}$ connected to $u_1$, $u_{20}$, $u_{22}$, $u_{42}$. 

The first 22 vectors $u_1,\ldots,u_{22}$ of $G_{\text{ind}}$ are chosen as above with $x = 1$ and $y = 2 + \sqrt{2}$. 
The 22 vectors $u'_1,\ldots,u'_{22}$ of $G'_{\text{ind}}$ are also obtained from the above solution, but with $0<x\leq 1$ a free parameter, and after applying first a unitary $U$ that maps  $(1,0,0)$ to $(1,-1,0)/\sqrt{2}$ and $(0,1,0)$ to $(1,1,0)\sqrt{2}$ and leave invariant $(0,0,1)$. We thus have $|u_1\rangle=|v_1\rangle=(1,0,0)^T$ and  $|u_{42}\rangle=|v_2\rangle=(1+x,1-x,0)^T/\sqrt{2+2x^2}$ as assumed.

By construction, the orthogonality relations of the subgraphs $G_{\text{ind}}$ and $G'_{\text{ind}}$ are satisfied. We also have that the vectors common to the two subgraphs are indeed identical, namely $|u_{20}\rangle=(1,-1,0)^T$ and $|u_{22}\rangle=(1,1,0)^T$. Furthemore, choosing $|u_{43}\rangle=(0,0,1)^T$, we also have that $|u_{43}\rangle$ is orthogonal to $|u_1\rangle, |u_{20}\rangle, |u_{22}\rangle$, and $|u_{42}\rangle$ as required. 

This completes the construction of the gadget for $\mathbb{C}^3$. Now, one may simply consider the same set of vectors as being embedded in any $\mathbb{C}^d$ (with additional vectors $(0,0,0,1,0,\dots,0)^T$, $(0,0,0,0,1,0,\dots,0)^T$ etc.) to construct a gadget in this dimension.
	\end{proof}
\begin{figure}[t]
	\centerline{\includegraphics[scale=0.38]{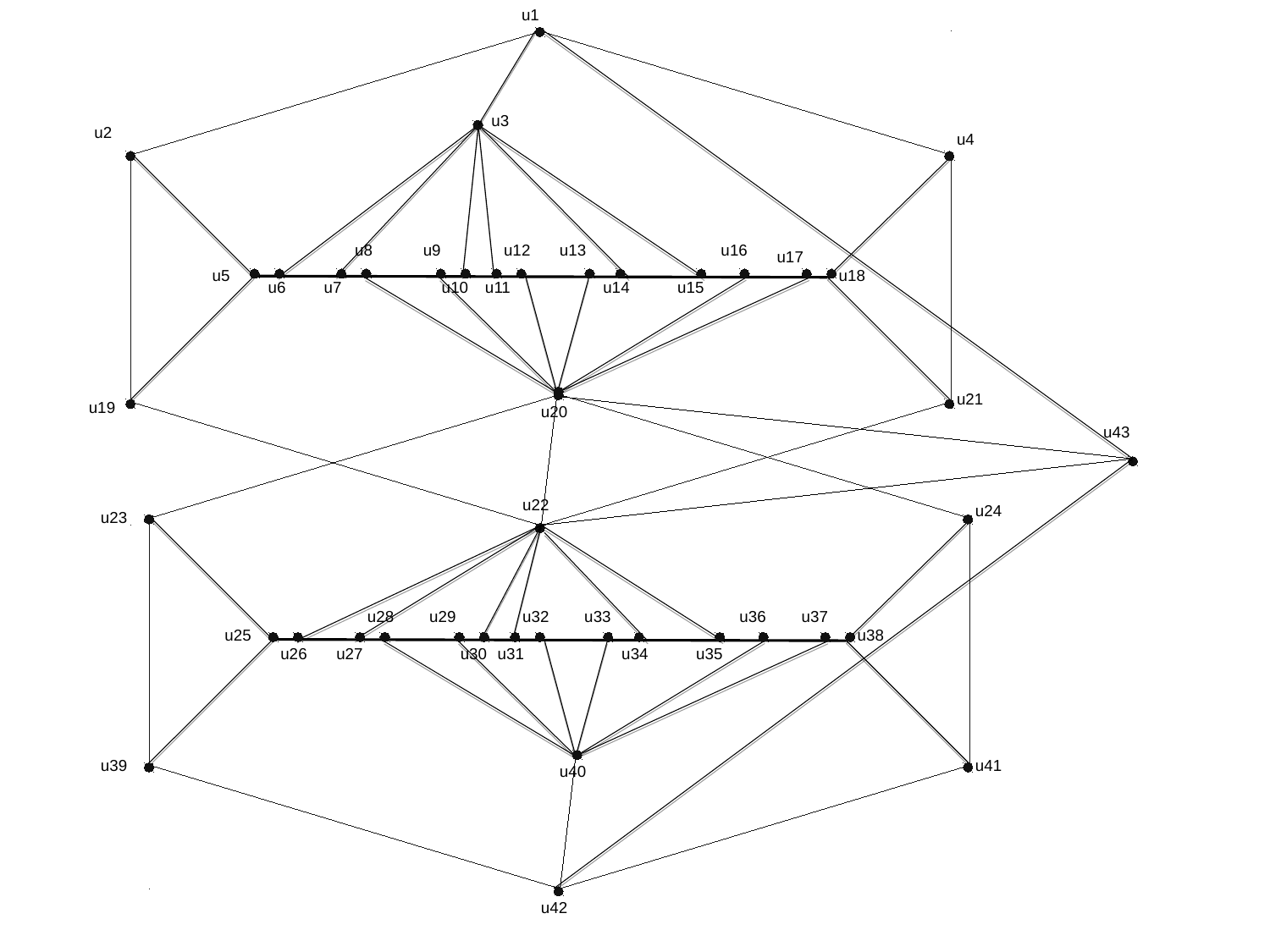}}
	\caption{The 43 vertex $01$-gadget used in the proof of Theorem~\ref{prop:fin-gadg-const}.}
	\label{fig:gadg-id-vec}
\end{figure}

Theorem \ref{prop:fin-gadg-const} allows to construct new KS graphs than the one given in the proof of Theorem \ref{prop:KS-gadg}. Some of such constructions in dimension 3 are shown in Fig. \ref{fig:KS-proofs}. A crucial role in these is played by the repeating unit $G_0$ shown in Fig. \ref{fig:KS-proofs} (a). This unit is given by a set of basis vectors $\{|u_1 \rangle, |u_2 \rangle, |u_3 \rangle\}$ all connected via appropriate $01$-gadgets to a central vector $|v_1 \rangle$. In any $\{0,1\}$-coloring $f$ of $G_0$, one of the three basis vectors must be assigned the value $1$, so that we necessarily have $f(|v_1 \rangle) = 0$. In other words, $G_0$ is a graph in which a particular vector necessarily takes value $0$ in any $\{0,1\}$-coloring. Note that this property is also shown by the graph in Fig. \ref{fig:gadg-id-vec} 

Note that from $G_0$, one can also construct an orthogonality graph $G_1$ in which a particular vector necessarily takes values 1 in any $\{0,1\}$-coloring. Indeed, consider two copies of $G_0$ with the respective central vectors $|v_1 \rangle$ and $|v_2 \rangle$ orthogonal to each other, so that $f(|v_1 \rangle) = f(|v_2 \rangle) = 0$. Then, in any $\{0,1\}$-coloring of the resulting graph $G_1$, the third basis vector $|v_3 \rangle \perp |v_1 \rangle, |v_2 \rangle$ necessarily obeys $f(|v_3 \rangle) = 1$.  

\begin{figure}
	\centerline{\includegraphics[scale=0.35]{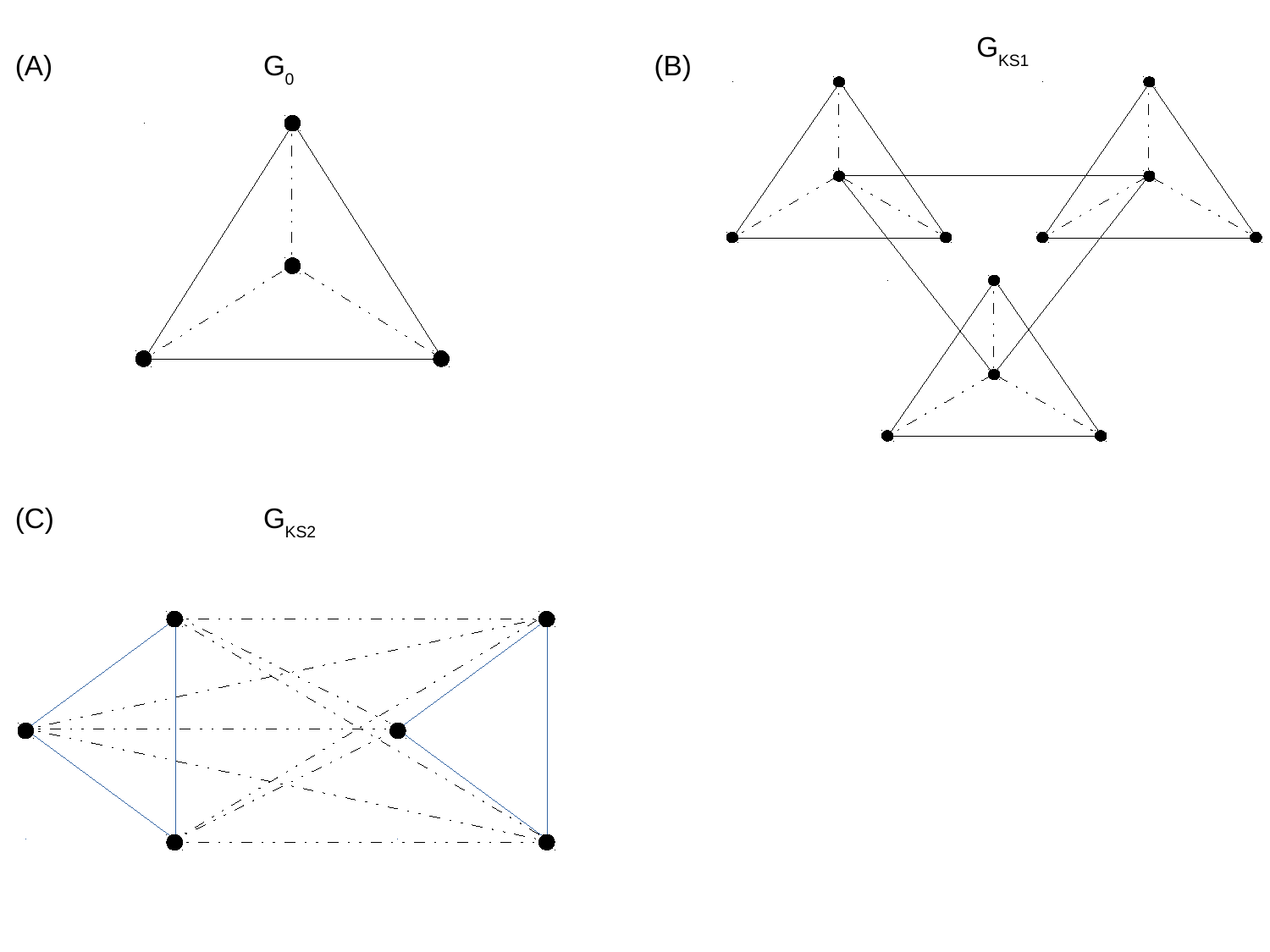}}
	\caption{Graphs with the dashed edges denoting $01$-gadgets. (a) In any $\{0,1\}$-coloring of the graph $G_0$, the central vertex is necessarily assigned value $0$. (b) Three copies of $G_0$ with the central vertices forming a basis in $\mathbb{C}^3$ so that the resulting graph $G_{KS1}$ forms a Kochen-Specker proof. (c) Another proof of the KS theorem $G_{KS2}$ is obtained by connecting every pair of vectors in two bases by a $01$-gadget.}
	\label{fig:KS-proofs}
\end{figure} 

In Fig. \ref{fig:KS-proofs} (b), a KS proof in $\mathbb{C}^3$ is based on the unit $G_0$, repeated three times with a basis set of central vectors $|v_1 \rangle, |v_2 \rangle, |v_3 \rangle$. By the property of $G_0$ in any $\{0,1\}$-coloring, all these three basis vectors are assigned value $0$ leading to a KS contradiction.
 In Fig. \ref{fig:KS-proofs} (c), the construction is based on two basis sets $\{|u_1 \rangle, |u_2 \rangle, |u_3 \rangle\}$ and $\{|v_1 \rangle, |v_2 \rangle, |v_3 \rangle\}$ with an appropriate $01$-gadget connecting every pair $|u_i \rangle, |v_j \rangle$ for $i, j =1, 2, 3$. So that assigning value $1$ to any of the vectors in one basis, necessarily implies that all of the vectors in the other basis are assigned value $0$, leading to a contradiction. Furthermore, the construction can be readily extended to derive KS graphs using any frustrated graph. 

\section{Statistical KS arguments based on $01$-gadgets}\label{sec:real}
The KS theorem can be seen as a proof that no non-contextual deterministic hidden-variable interpretation of quantum theory is possible. In a deterministic hidden-variable model, we aim to reproduce the quantum probabilities
\begin{equation}
\text{Pr}_{\psi}(i|M)=\sum_\lambda q_{\psi}(\lambda) f_\lambda(i|M)
\end{equation}
in term of hidden-variables $\lambda$, where a distribution $q_{\psi}(\lambda)$ over the hidden-variables is associated to each quantum state $|\psi\rangle$, and where for each $\lambda$, the model predicts with certainty that one of the outcomes $i$ will occur for each measurement $M$, i.e., the hidden measurement outcome probabilities $f_\lambda(i|M)$ satisfy $f_\lambda(i|M)\in\{0,1\}$. Furthermore, the model is non-contextual if, as in the quantum case, the probabilistic assignment to the outcome $i$ of the (projective) measurement $M$, only depends on the corresponding projector $V_i$, independently of the wider context provided by the full description of the measurement $M=\{V_1,V_2,\ldots,V_n\}$. In other words in a non-contextual deterministic hidden-variable, we aim to write for every projector $V$:
\begin{equation}
\langle\psi|V|\psi\rangle=\sum_\lambda q_{\psi}(\lambda) f_\lambda(V)\,,
\end{equation}
where $f_\lambda(V)\in\{0,1\}$. Obviously, we should also require for consistency that $\sum_{i\in \mathcal{O}}f(V_i)\leq 1$ for any set $\mathcal{O}$ of mutually orthogonal projectors, with equality when the projectors in $\mathcal{O}$ sum to the identity. 

No-go theorems against such models, i.e., ``proofs of contextuality" , are usually obtained by considering a finite set $\mathcal{S}=\{|v_1\rangle,\ldots,|v_n\rangle\} \subset \mathbb{C}^d$ of rank-one projectors $V_i$, represented as vectors through $V_i=|v_i\rangle\langle v_i|$. Specializing to this case, a non-contextual hidden variable model should satisfy for each $|v_i\rangle$ in $\mathcal{S}$ and each $|\psi\rangle$ in $\mathbb{C}^d$,
\begin{equation}\label{eq:nchv}
|\langle\psi|v_i\rangle|^2=\sum_\lambda q_{\psi}(\lambda) f_\lambda(|v_i\rangle)\,,
\end{equation}
where the $f_\lambda:\mathcal{S}\rightarrow \{0,1\}$ are $\{0,1\}$-colorings of $\mathcal{S}$.

At least three types of no-go theorems, from strongest to weakest, against such non-contextual hidden-variable models can be constructed.

The first types correspond to Kochen-Specker theorems. They establish that for certain sets $\mathcal{S}$, it is not possible to consistently define $\{0,1\}$-colorings $f_\lambda$ of $\mathcal{S}$, even before attempting to use them to reproduce the quantum probabilities. This is what we have discussed until now. 

In the second type of proofs, a $\{0,1\}$-coloring of $\mathcal{S}$ is not excluded. But it can be shown that for any such coloring $f_\lambda$ of $\mathcal{S}$, a certain inequality  $\sum_i c_i f_\lambda(|v_i\rangle)\leq c_0$ must necessarily be satisfied, while in the quantum case, it happens that $\sum_i c_i |v_i\rangle\langle v_i|>c_0 \mathbb{I}$. In other words, though it is possible to find a $\{0,1\}$ assignment $f_\lambda(|v_i\rangle)$ to each projector $|v_i\rangle\langle v_i|$ in $\mathcal{S}$ that is compatible with the orthogonality relations among such projectors, any such assignment fails to reproduce some more complex relation of the type $\sum_i c_i |v_i\rangle\langle v_i|>c_0 \mathbb{I}$ satisfied by these projectors. This immediately implies a contradiction with eq.~(\ref{eq:nchv}), since in the quantum case we have for any $|\psi\rangle$, $\sum_i c_i |\langle\psi|v_i\rangle|^2>c_0$, while according to a non-contextual hidden variable model, we would have $\sum_i c_i |\langle\psi|v_i\rangle|^2=\sum_\lambda q_{\psi}(\lambda)\left[\sum_i c_i f_\lambda(|v_i\rangle)\right]\leq \sum_\lambda q_{|\psi\rangle}(\lambda)c_0 \leq c_0$. Such no-go theorems are referred to as ``statistical state-independent" KS arguments and were introduced by Yu and Oh \cite{YO12}.

Finally, for certain sets $\mathcal{S}$, it is possible to find valid $\{0,1\}$-colorings that do not lead to any type of contradictions of the second type above. However, it is not possible to take mixtures of such colorings, as in eq.~(\ref{eq:nchv}), to reproduce the predictions of certain quantum states $|\psi\rangle$. Such no-go theorems are referred to as ``statistical state-dependent" KS arguments and were introduced by Clifton in \cite{Clifton93}.

While we have seen in the previous section how proofs of the KS theorem can be constructed using $01$-gadgets, in this section we show how to use them to build statistical state-independent and state-dependent KS arguments

\subsection{State-independent KS arguments}
In \cite{YO12}, Yu and Oh introduced a set of 13 vectors in $\mathbb{C}^ 3$ that provides a state-independent proof of contextuality, despite not being a KS set. We show how using Theorem~\ref{prop:fin-gadg-const}, it is possible to construct other state-independent proofs of contextuality based on $01$-gadgets.  To do this, we make use of the following lemma.

\begin{lemma}
	\label{lem:d-simplex}
	Let $|u_i \rangle$, for $i =1, \dots, d+1$ be the unit vectors denoting the vertices of a $d$-dimensional simplex embedded in $\mathbb{R}^d$. Then 
	\begin{eqnarray}
	\sum_{i=1}^{d+1} |u_i \rangle \langle u_i | = \frac{d+1}{d} \mathbb{I}.
	\end{eqnarray}
\end{lemma}
\begin{proof}
Since $|u_i \rangle$ form the vertices of the $d$-simplex, we have $ \langle u_i | u_j \rangle = -\frac{1}{d}$ for any $i \neq j \in \{1, \dots, d+1\}$. It then follows 
	\begin{eqnarray}
	\left(\sum_{i=1}^{d+1} \langle u_i | \right) \left(\sum_{j=1}^{d+1} | u_i \rangle \right) = (d+1) + d(d+1)\left(-\frac{1}{d} \right) = 0, \nonumber
	\end{eqnarray}
	so that 
	\begin{eqnarray}
	O := \sum_{i=1}^{d+1} |u_i \rangle \langle u_i| = -\sum_{i \neq j = 1}^{d+1} | u_i \rangle \langle u_j|
	\end{eqnarray}
	This then implies that
	\begin{eqnarray}
	O^{2} = O - \frac{1}{d} \sum_{i \neq j =1}^{d+1} |u_i \rangle \langle u_j| = \frac{d+1}{d} O.
	\end{eqnarray}
    Moreover, $O$ is invertible, since $\text{span}(\{|u_i \rangle\}_{i=1}^{d+1}) = \mathbb{R}^d$ so that we obtain $O = \frac{d+1}{d} \mathbb{I}$.
\end{proof}
Now, state-independent KS arguments for $\mathbb{C}^d$ are straightforwardly constructed as follows. For every pair of vectors $|u_i \rangle, |u_j \rangle$ of the $d$-simplex, consider a $01$-gadget $\mathcal{S}_{ij}$ with $|u_i\rangle$, $|u_j\rangle$ the distinguished vertices. Since $|u_i\rangle$ and $|u_j\rangle$ are non-orthogonals, such gadgets exists, as implied by Theorem~\ref{prop:fin-gadg-const}. The resulting set of vectors $\mathcal{S}=\cup_{ij} \mathcal{S}_{ij}$ exhibits state-independent contextuality. Indeed, by the property of the $01$-gadgets, only one of the vectors $|u_i \rangle$ for $i =1, \dots, d+1$ can be assigned the value $1$ in any $\{0,1\}$-coloring of $\mathcal{S}$. It thus follows that
\begin{eqnarray}
\sum_{i=1}^{d+1} f(|u_i \rangle) \leq 1,.
\end{eqnarray}
On the other hand, from Lemma \ref{lem:d-simplex}, every state $|\psi\rangle$ from $\mathbb{C}^d$ achieves the value $\sum_{i=1}^{d+1} |\langle\psi|u_i\rangle|^2=\frac{d+1}{d}>1$.

While we have used the $d+1$ vertices of a $d$-simplex in the construction above, we observe that any set $\{|u_i \rangle\}$ of vectors in $\mathbb{C}^d$ such that $\sum_{i} |\langle \psi | u_i \rangle|^2 > 1$ for all $| \psi \rangle \in \mathbb{C}^d$ can be utilized in the construction, although such a set clearly needs to contain at least $d+1$ vectors. 

\subsection{State-dependent KS arguments}

The relation between state-dependent KS arguments and 01-gadgets is even more direct than in the above construction. Actually, the first state-dependent KS argument introduced by Clifton in \cite{Clifton93} was precisely based on the set of vectors (\ref{eq:Clif-orth-rep}) forming the Clifton gadget $G_\text{gad}$. His argument was as follows. In every non-contextual hidden-variable model attempting to replicate the quantum probabilities associated to the projectors of the Clifton gadget, we should have $|\langle\psi|u_1\rangle|^2+|\langle\psi|u_8\rangle|^2=\sum_\lambda q_\psi(\lambda) \left(f_\lambda(|u_1\rangle)+f_\lambda(|u_8\rangle)\right)\leq 1$, by the gadget property. However, if we take $|\psi\rangle=|u_1\rangle$, we find that according to the quantum predictions $|\langle u_1|u_1\rangle|^2+|\langle u_1|u_8\rangle|^2=1+|\langle u_1|u_8\rangle|^2>1$ since $|\langle u_1|u_8\rangle|^2>0$ as $|u_1\rangle$ and $|u_8\rangle$ are non-orthogonal. Other state-dependent proofs based on inequalities have since been developed, with the smallest involving five vectors \cite{KCBS08}. The first state-independent statistical KS argument was presented in \cite{Cab08} and the proof that any KS set give can be converted in a state-independent statistical KS argument was presented in \cite{BBCP09}.

Obviously, the argument used by Clifton for the particular set of vectors he introduced, immediately carries over to any $01$-gadget. Thus every $01$-gadget serves as a proof of state-dependent contextuality.

Note that it was realized in \cite{CDLP14} that a class of graphs, known as perfect graphs, define a class of graphs that cannot serve as proofs of (even state-dependent) contextuality. That is, for any orthogonal representation $\{|v_j \rangle\} \subset \mathbb{C}^d$ of a perfect graph and for any pure state $|\psi \rangle \in \mathbb{C}^d$, the outcome probabilities $|\langle \psi|v_j\rangle|^2$ admit a non-contextual hidden variable model of the form (\ref{eq:nchv}). Since a non-contextual hidden variable model is not possible for a $01$-gadget, we deduce that no perfect graph is a $01$-gadget. Perfect graphs are a well-known class of graphs which by the strong perfect graph theorem \cite{CRST06} can be characterized as those graphs that do not contain odd cycles and anti-cycles of length greater than three as induced subgraphs. 

Finally, remark that the argument due to Clifton presented above works not only for the state $|\psi\rangle=|u_1\rangle$, but for any state $|\psi \rangle \in \mathbb{C}^3$ which obeys $| \langle \psi | u_1 \rangle|^2 + |\langle \psi | u_8 \rangle|^2 > 1$. More generally, we now present a $01$-gadget which serves to prove state-dependent contextuality for all but a measure zero set of states in $\mathbb{C}^3$. 

This construction is based on the gadget $G$ of Fig.~\ref{fig:gadg-id-vec} with the 43 vector orthogonal representation presented in the proof of Theorem~\ref{prop:fin-gadg-const}. Note that if we take $x=1$ in this representation, then the two distinguished vectors $|u_1\rangle$ and $|u_{42}\rangle$ actually coincide and are both equal to $(1,0,0)$ (i.e., the two distinguished vertices $u_1$ and $u_{42}$ should actually be identified). Therefore in any $\{0,1\}$-coloring $f$ of $G$, $2f(|u_1\rangle)=f(|u_1\rangle)+f(|u_{42}\rangle)\leq 1$, i.e. the vector $|v_1\rangle$ is assigned value $0$. This implies that $G$ witnesses state-dependent contextuality of all states in $\mathbb{C}^3$ but for a measure zero set of states $|\psi \rangle$ that are orthogonal to $|v_1\rangle=(1,0,0)$.

The construction that we just described is based on 42 vectors. It is actually possible to find a slightly smaller construction based on the following 40 vectors:
\begin{eqnarray}
&&|u_1 \rangle = (1,-1,0)^T; \; \; |u_2 \rangle = (1,1,1)^T;\nonumber \\
&&|u_3 \rangle = (1,1,0)^T; \; \; |u_4 \rangle = (1,1,b)^T; \nonumber \\
&&|u_5 \rangle = (-2,1,1)^T; \; \; |u_6 \rangle = (1,-1,3)^T; \nonumber \\
&&|u_7 \rangle = (3,-3,-2)^T; \; \; |u_8 \rangle = (2,0,3)^T;\nonumber \\
&& |u_9 \rangle = (-3,0,2)^T; \;\;|u_{10} \rangle = (-2,2,-3)^T; \nonumber \\
&& |u_{11} \rangle = (3,-3,-4)^T; \; \; |u_{12} \rangle = (4,0,3)^T; \nonumber \\
&&|u_{13} \rangle = (-3,0,4)^T; \; \; |u_{14} \rangle = (-4,4,-3)^T;  \nonumber \\
&&|u_{15} \rangle = (3,-3,-8)^T; \; \; |u_{16} \rangle = (8,0,3)^T; \nonumber \\
&& |u_{17} \rangle = (-3,0,8)^T; \;\;|u_{18} \rangle = (-8,4+\sqrt{7},-3)^T;\nonumber \\
&&|u_{19} \rangle = (0,1,-1)^T; |u_{20} \rangle = (0,1,0)^T; \nonumber \\
&& |u_{21} \rangle = (0,-3+8b,-16-3b)^T; \;\;|u_{22} \rangle = (1,0,0)^T; \nonumber \\
&&|u_{23} \rangle = (1,0,-1)^T; \;\; |u_{24} \rangle = (2-\sqrt{2},0,1)^T; \nonumber \\
&&|u_{25} \rangle = (1,-2,1)^T; \; \; |u_{26} \rangle = (0,1,2)^T; \nonumber \\
&& |u_{27} \rangle = (0,2,-1)^T; \;\; |u_{28} \rangle = (1,-1,-2)^T; \nonumber \\
&& |u_{29} \rangle = (1,-1,1)^T; \; \; |u_{30} \rangle = (0,1,1)^T; \nonumber \\
&&|u_{31} \rangle = (0,1,-1)^T; \;\; |u_{32} \rangle = (-1,1,1)^T; \nonumber \\
&&|u_{33} \rangle = (-1,1,-2)^T; \;\; |u_{34} \rangle = (0,2,1)^T; \nonumber \\
&& |u_{35} \rangle = (0,1,-2)^T; \; \; |u_{36} \rangle = (2,-2,-1)^T; \nonumber \\
&&|u_{37} \rangle = (1,-1,4)^T; \; \;  |u_{38} \rangle = (-2-\sqrt{2},6-\sqrt{2},2)^T;  \nonumber \\
&& |u_{39} \rangle = |u_2 \rangle;\; \; |u_{40} \rangle = |u_3 \rangle; \; \; |u_{41} \rangle = (1,1,-2+\sqrt{2})^T; \nonumber \\
&& |u_{42} \rangle = |u_1 \rangle; |u_{43} \rangle = (0,0,1)^T; \nonumber 
\end{eqnarray}
with $b = \frac{-4+\sqrt{7}}{3}$, and where we have the following identities $|u_{1}\rangle=|u_{42}\rangle$, $|u_2\rangle=|u_{39}\rangle$,  $|u_3\rangle=|u_{40}\rangle$. It can be verified that the graph in Fig.~\ref{fig:gadg-id-vec} where we identify the vertices $u_{1}$ and $u_{42}$, $u_2$ and $u_{39}$,  $u_3$ and $u_{40}$, is the orthogonality graph of these 40 vectors. These 40 vectors thus form a $01$-gadget, where as above the vector $|u_1\rangle=(1,-1,0)$ can only be assigned the value $0$, implying that it can serve as a state-dependent contextuality proof for any vector in $\mathbb{C}^3$ that is not orthogonal to $(1,-1,0)$. We leave it as an open question whether this set of $40$ vectors is the minimal set with this property.

\section{Proofs of the extended Kochen-Specker theorem using $01$-gadgets}\label{sec:ext}
In this section, we consider a stronger variant of the KS theorem due to Pitowsky \cite{Pitowsky} and Hrushovski and Pitowsky \cite{HP03}. While the KS theorem is concerned with $\{0,1\}$-colorings where all projectors (or vectors) in a given set $S$ must be assigned a value in $\{0,1\}$, we consider here more general assignments where any real value in $[0,1]$ is allowed to the members of $S$. Specifically, given a set of vectors $\mathcal{S}=\{|v_1\rangle,\ldots,|v_n\rangle\}\subset\mathbb{C}^d$, we say that $f:\mathcal{S}\rightarrow [0,1]$ is a $[0,1]$-assignment if $f$ satisfies the same rules (\ref{eq:01rule}) as it does for $\{0,1\}$-colorings.
Both $\{0,1\}$-colorings and $[0,1]$-assignments can be interpreted as assigning a probability to the projectors corresponding to each of the elements of $S$. But while the assignment is constrained to be deterministic in the case of $\{0,1\}$-colorings since these probabilities can only take the values $0$ or $1$, the probabilistic assignment may be completely general (hence non-deterministic) for $[0,1]$-assignments. In particular, for any given quantum state $|\psi\rangle$, the Born rule $f(|v_i\rangle)=|\langle \psi|v_i\rangle|^2$ defines a valid $[0,1]$-assignment.

Hrushovski and Pitowsky \cite{HP03}, following earlier work by Pitowsky in \cite{Pitowsky}, proved the following theorem, which they call the ``logical indeterminacy principle".
\begin{theorem}[\cite{HP03}]
	\label{thm:HP03}
	Let $|v_1 \rangle$ and $|v_2 \rangle$ be two non-orthogonal vectors in $\mathbb{C}^d$ with $d \geq 3$. Then there is a finite set of vectors $S \subset \mathbb{C}^d$ with $|v_1 \rangle, |v_2 \rangle \in S$ such that 
	for any $[0,1]$-assignment, it holds that $f(|v_1 \rangle), f(|v_2 \rangle) \in \{0, 1\}$ if and only if $f(|v_1\rangle) = f(|v_2\rangle) = 0$. 
\end{theorem}
Thus for any two non-orthogonal vectors $|v_1\rangle$ and $|v_2\rangle$, at least one of the probabilities associated to the  vectors $|v_1\rangle$ or $|v_2\rangle$ must be strictly between zero and one, unless they are both equal to zero. A corollary of this result, observed in \cite{ACCS12, ACS14, ACS14-2} is that if $f(|v_1 \rangle) = 1$ (this should, for instance, necessarily be the case if we attempt to reproduce the quantum probabilities for measurements performed on the state $|\psi\rangle=|v_1\rangle$), then $f(|v_2 \rangle) \neq 0,1$, showing that one can localise the ``value-indefiniteness" of quantum observables that the KS theorem implies. Theorem~\ref{thm:HP03} therefore provides a stronger variant of the KS theorem, and we will refer to it as the \emph{extended KS theorem}. 

The proof of Theorem \ref{thm:HP03} given in \cite{HP03} was obtained as a corollary of Gleason's theorem \cite{Gleason}.
A more explicit constructive proof was given by Abbott, Calude and Svozil \cite{ACCS12, ACS14}, where they also noted that significantly none of the known KS sets serves to prove Theorem~\ref{thm:HP03}. Note that an earlier proof of the extended KS theorem was also given in \cite{Pitowsky}. All these existing proofs of the extended KS theorem involve complicated constructions with no systematic procedure for obtaining the requisite sets of vectors. In this subsection, we will provide a simple systematic method for obtaining in a constructive way these extended KS sets. 

In order to prove the extended KS theorem, we need gadgets of a special kind, which are defined as  usual $01$-gadgets apart from the fact that the condition that the two distinguished vertices cannot both be assigned the value $1$ in any $\{0,1\}$-colorings should also hold for any $[0,1]$-assignments. That is, we simply replace `$\{0,1\}$-coloring' by `$[0,1]$-assignment' and $f(|v_1\rangle)+f(|v_2\rangle)\leq 1$ by $f(|v_1\rangle)+f(|v_2\rangle)<2$ in Definition 1, and similarly for Definition 2. We call such new gadgets `extended $01$-gadgets'.  It is easily verified that the Clifton gadget in Fig. \ref{fig:Clifton} and the $16$-vertex gadget in Fig. \ref{fig:cab-KS-gadget} obey this additional restriction. 

Our first aim will be to construct such extended $01$-gadgets for any two given non-orthogonal vectors $|v_1 \rangle, |v_2 \rangle \in \mathbb{C}^d$ for $d \geq 3$. This is the content of the following Theorem, which generalizes Theorem~\ref{prop:fin-gadg-const}.

\begin{theorem}
	\label{prop:101-gadg-const}
	Let $|v_1 \rangle$ and $|v_2 \rangle$ be any two distinct non-orthogonal vectors in $\mathbb{C}^d$ with $d \geq 3$. Then there exists an extended 01-gadget in dimension $d$ with $|v_1 \rangle$ and $|v_2 \rangle$ being the two distinguished vertices.
\end{theorem}
\begin{proof}
	We begin with the construction for $d=3$ and generalize it to higher dimensions naturally. The construction is an iterative procedure based on the Clifton gadget $G_{\text{Clif}}$ given in Fig. \ref{fig:Clifton}.
	
Firstly, as stated previously, it is readily seen that $G_{\text{Clif}}$ is actually an extended $01$-gadget with $u_1, u_8$ the two distinguished vertices, i.e., any $[0,1]$-assignment $f : V(G_{\text{Clif}}) \rightarrow [0,1]$ cannot be such that $f(u_1)=f(u_8)=1$. 
Further, it is known that the $\mathbb{R}^3$ realization of $G_{\text{Clif}}$ given by (\ref{eq:Clif-orth-rep})  achieves the (minimal possible) separation of $\theta_1 = \arccos{|\langle u_1 | u_8 \rangle|} = \arccos{1/3}$ between the two end vertices \cite{Stanford}. 

\begin{figure}
	\centerline{\includegraphics[scale=0.23]{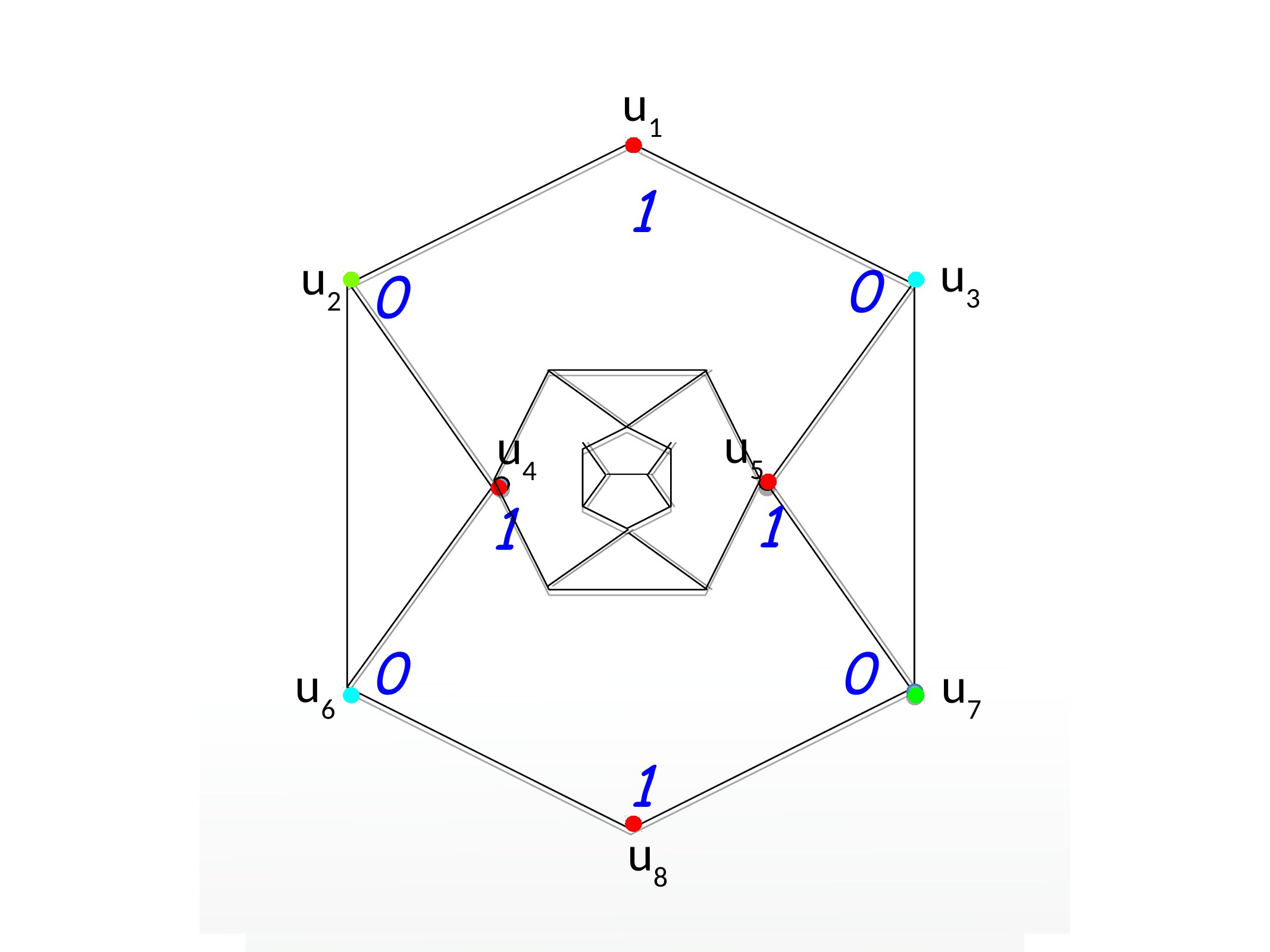}}
	\caption{An iterative construction of an extended $01$-gadget for which the two distinguished vertices $u_1$ and $u_8$ are such that in the limit of large number of iterations $k$, $ |\langle u_1^{(k)} | u_8^{(k)}\rangle|\in [0,1[$.}
	\label{fig:Gad-construct}
\end{figure} 

We now describe a nesting procedure that at each step decreases the angle between the vectors corresponding to the two outer vertices. The procedure works as follows. Replace the edge $(u_4, u_5)$ in $G_{\text{Clif}}$ by $G'_{\text{Clif}}$, a copy of $G_{\text{Clif}}$ where we identify $u'_1=u_4$ and $u'_8=u_5$. The new graph thus obtained has $14$ vertices and $21$ edges. The operation has the property that in any $[0,1]$-assignment $f$, an assignment of value $1$ to the two outer vertices of the new graph (i.e. $u_1,u_8$) leads to a similar assignment to the two outer vertices of the inner copy of $G_{\text{Clif}}$ (i.e. $u'_1,u'_8$) thereby giving rise to a contradiction. In other words, the newly constructed graph is once again an extended $01$-gadget.  This procedure can be repeated an arbitrary number of times, as illustrated in Fig. \ref{fig:Gad-construct}, leading to an extended $01$-gadget formed from $k$ nested Clifford graphs $G_{\text{Clif}}^1,G_{\text{Clif}}^2,G_{\text{Clif}}^2,\ldots,G_{\text{Clif}}^k$ where $G_{\text{Clif}}^1$ corresponds to the most inner graph and $G_{\text{Clif}}^k$ to the most outer graph.
We now show that the total graph at the $k$-th iteration is an orthogonality graph where the overlap $|\langle u_1^{(k)} | u_8^{(k)} \rangle |$ between the two outer vertices $u_1^k,u_8^k$ can be chosen to take any value in $[0,\frac{k}{k+2}]$, thus spanning any possible value in $[0,1[$ for $k$ sufficiently large. 
Setting $|v_1 \rangle = |u_1^{(k)} \rangle$ and $|v_2 \rangle = |u_8^{(k)} \rangle$ with $k$ depending on the overlap of the given vectors $| \langle v_1 | v_2 \rangle|$, then gives the required gadget and proves the Theorem.

Suppose that at the $k$-th step of the iteration, the vectors representing the two outer vertices of the ``inner" gadget from the $k-1$-th step are 
\begin{eqnarray}
&&|u_4^{(k)} \rangle = |u_1^{(k-1)} \rangle = (1,0,0), \nonumber \\
&&|u_5^{(k)} \rangle = |u_8^{(k-1)} \rangle = \frac{1}{\sqrt{1+x_k^2}}(x_k,1,0),
\end{eqnarray} 
without loss of generality, so that the overlap between these vectors is $| \langle u_4^{(k)} | u_5^{(k)} \rangle| = \frac{x_k}{\sqrt{1+x_k^2}}$, where for simplicity of the construction we take $x_k \in \mathbb{R}^+_0$. The remaining vectors then in general have the following (non-normalized) orthogonal representation in $\mathbb{R}^3$
\begin{eqnarray}
&&|u_8^{(k)} \rangle = (a_k,b_k,c_k), \; \; |u_6^{(k)} \rangle = (0,-c_k,b_k), \nonumber \\
&& |u_7^{(k)} \rangle = (c_k,-c_k x_k, -a_k + b_k x_k), \; \; |u_2^{(k)} \rangle = (0,b_k,c_k), \; \;  \nonumber \\
&&|u_3^{(k)} \rangle = (-a_k + b_k x_k, a_k x_k - b_k x_k^2, -c_k -c_k x_k^2), \nonumber \\
&& |u_1 \rangle = (-b_k c_k- a_k c_k x_k, - a_k c_k + b_k c_k x_k,a_k b_k - b_k^2 x_k), \nonumber \\
\end{eqnarray}
with $a_k,b_k,c_k \in \mathbb{R}$. This gives an overlap of
\begin{widetext}
\begin{eqnarray}
\label{eq:overlap}
| \langle u_1^{(k)} | u_8^{(k)} \rangle | = \frac{|-a_k c_k (b_k + a_k x_k)|}{\sqrt{(a_k^2 + b_k^2 + c_k^2)(c_k^2(b_k + a_k x_k)^2 + b_k^2(a_k - b_k x_k)^2 + (a_k c_k - b_k c_k x_k)^2)}}.
\end{eqnarray}
\end{widetext}
A direct optimization of this expression with respect to the parameters $a_k,b_k,c_k$ gives the choice $b_k=1$, $c_k=1$, $a_k= x_k + \sqrt{1+x_k^2}$. So that the overlap between the two outer vertices at the $k$-th step of the iteration is given by 
\begin{eqnarray}
\label{eq:iter-koverlap}
| \langle u_1^{(k)} | u_8^{(k)} \rangle | = \frac{1}{3+4 x_k(x_k-\sqrt{1+x_k^2})} =: \frac{x_{k+1}}{\sqrt{1+x_{k+1}^2}}.
\end{eqnarray}
With the initial overlap for $k=1$ of $1/3$ and corresponding initial $x$ values of $x_1 = 0$ and $x_2=\frac{1}{2 \sqrt{2}}$, we can now evaluate the expression for the overlap for any $k>1$. We find that the overlap at the $k$-th step is $\frac{k}{k+2}$. This is readily seen by an inductive argument. The base claim is clear, suppose that at the $k$-th step the overlap is given by $\frac{x_{k+1}}{\sqrt{1+x_{k+1}^2}} = \frac{k}{k+2}$, i.e., $x_{k+1} = \frac{k}{2\sqrt{k+1}}$. Substituting in Eq.\ref{eq:iter-koverlap}, we obtain $\frac{x_{k+2}}{\sqrt{1+x_{k+2}^2}} = \frac{k+1}{k+3}=\frac{(k+1)}{(k+1)+2}$. Moreover, we see that choosing $b_k=1, c_k=1$, the overlap expression (\ref{eq:overlap}) is a continuous function of $a_k$ for any fixed $x_k$ with the minimum value of $0$ achieved at $a_k = 0$. Thus, every intermediate overlap in $[0,\frac{k}{k+2}]$ between the two outer vectors is also achievable by appropriate choice of $a_k$ for the fixed value of $x_k, b_k, c_k$. This completes the construction of the gadget for $\mathbb{C}^3$ (possibly by taking its faithful version in the graph representation).

Now, one may simply consider the same set of vectors as being embedded in any $\mathbb{C}^d$ (with additional vectors$(0,0,0,1,0,\dots,0)^T$, $(0,0,0,0,1,0,\dots,0)^T$ etc.) to construct a gadget in this dimension.  
\end{proof}

\begin{figure}
	\centerline{\includegraphics[scale=0.38]{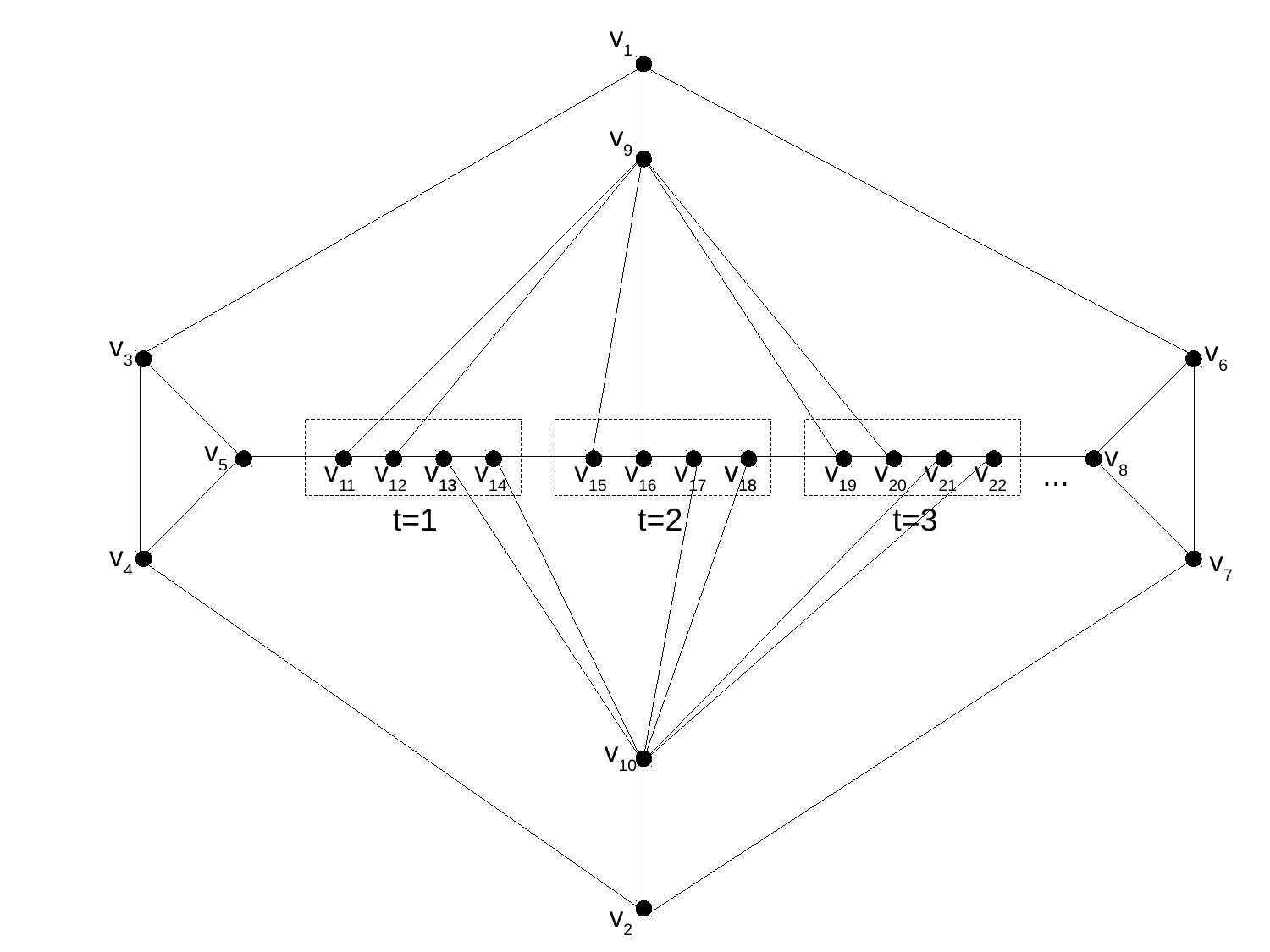}}
	\caption{An alternative construction of an extended $01$-gadget for which the two distinguished vertices $v_1$ and $v_2$ are such that in the limit of large number $t$ of the repeating unit of four vectors, $| \langle v_1 | v_2 \rangle|$ can take any value in $[0,1[$.}
	\label{fig:ext-gad}
\end{figure} 

In fact, the construction above is not unique. We give an alternative set of vectors that also serves to prove Theorem~\ref{prop:101-gadg-const}. The construction is shown in Fig. \ref{fig:ext-gad}. Suppose we are given two distinct non-orthogonal vectors $|v_1 \rangle = (1,0,0)^T$, $|v_2 \rangle = (x,\sqrt{1-x^2},0)^T$, with $0 < x < 1$. We begin by adding the following set of vectors with a parameter $y \in \mathbb{R}$:
	\begin{eqnarray}
&&|v_3 \rangle = (0,x,-\sqrt{1-x^2})^T; \nonumber \\
&&|v_4 \rangle = (-(1-x^2),x \sqrt{1-x^2},x^2)^T; \nonumber \\
&&|v_5 \rangle = (x,(1-x^2)\sqrt{1-x^2},x(1-x^2))^T; \nonumber \\
&&|v_6 \rangle = (0,y,\sqrt{1-y^2})^T; \nonumber \\
&&|v_7 \rangle = (-\sqrt{(1-x^2)(1-y^2)},x\sqrt{1-y^2}, x y)^T; \nonumber \\
&&|v_8 \rangle = (x,(1-y^2)\sqrt{1-x^2},y\sqrt{(1-x^2)(1-y^2)})^T; \nonumber \\
&&|v_9 \rangle = (0,1,0)^T; \; \; |v_{10} \rangle = (-\sqrt{1-x^2},x,0)^T.
\end{eqnarray}    
The remaining vectors are obtained using a repeating unit consisting of four vectors:
\begin{eqnarray}
&&|v_{7 + 4t} \rangle = (-(1-x^2),0,x^{2(t-1)})^T; \nonumber \\
&&|v_{8 + 4t} \rangle = (x^{2(t-1)},0,1-x^2)^T; \nonumber \\
&&|v_{9+4t} \rangle = (-x(1-x^2),-(1-x^2)\sqrt{1-x^2},x^{2t-1})^T; \nonumber \\
&&|v_{10+4t} \rangle = (x^{2t},x^{2t-1}\sqrt{1-x^2},1-x^2)^T; 
\end{eqnarray}
repeated $t$ times for an integer $t \geq 1$ depending on $x$. Choosing the parameter $y$ as 
\begin{eqnarray}
y = \sqrt{\frac{(1-x^2)^2+2 x^{4t-2} - \sqrt{(1-x^2)((1-x^2)^3 - 4 x^{4t})}}{2(1-x^2)(1-x^2+x^{4t-2})}}, \nonumber
\end{eqnarray}
we find that $y \in \mathbb{R}$, for $t$ satisfying $(1-x^2)^3 \geq 4 x^{4t}$. We see that as $t$ increases this inequality can be satisfied for larger values of $x$, and for any $0 < x < 1$ as $t \rightarrow \infty$. From the orthogonality graph of this set of vectors $S$ shown in Fig. \ref{fig:ext-gad}, it is clear that there cannot be any assignement $f: S \rightarrow [0,1]$ such that $f(|v_1 \rangle)=f(|v_2\rangle)=1$, giving an extended $01$-gadget.

While the construction in Theorem~\ref{prop:101-gadg-const} and that in the previous paragraph work for any two distinct vectors, given two such vectors it is of great interest to find the minimal extended $01$-gadget with these vectors as the distinguished vertices. While this question is the foundational analog for extended KS systems of the question of finding minimal KS sets, it is also of practical interest in obtaining Hardy paradoxes with optimal values of the non-zero probability, and extracting randomness from the gadgets \cite{R17}. 

We now show how the extended $01$-gadgets can be used to construct proofs of the extended KS Theorem~\ref{thm:HP03}.

\begin{proof}(Theorem \ref{thm:HP03})
We present the construction for $d=3$, the proof for higher dimensions will follow in an analogous fashion. The idea is encapsulated by Fig. \ref{fig:Ext-KS}. Suppose we are given two distinct non-orthogonal vectors $|v_1 \rangle$ and $|v_2 \rangle$ in $\mathbb{C}^d$. We begin by constructing an appropriate extended $01$-gadget $G_{v_1,v_2}$, depending on $| \langle v_1 | v_2 \rangle|$, with the corresponding $v_1, v_2$ being the distinguished vertices. \ 

Let $|v_3 \rangle = |v_1 \rangle \times |v_2 \rangle$ denote the vector orthogonal to the plane $\text{span}(|v_1 \rangle, |v_2\rangle)$ spanned by $|v_1 \rangle$ and $|v_2 \rangle$, where $\times$ denotes the cross product of the vectors. Let $|v_4 \rangle$ be the vector in the plane $\text{span}(|v_1 \rangle, |v_2\rangle)$ orthogonal to $|v_1 \rangle$, and $|v_5 \rangle$ denote the vector in this plane orthogonal to $|v_2 \rangle$, so that $\{|v_1 \rangle, |v_3 \rangle, |v_4 \rangle\}$, $\{|v_2 \rangle, |v_3 \rangle, |v_5 \rangle\}$ form orthogonal bases in $\mathbb{C}^3$. We construct appropriate extended $01$-gadgets $G_{v_1, v_5}$ and $G_{v_2, v_4}$ depending on $| \langle v_1 | v_5 \rangle|$ and $| \langle v_2 | v_4 \rangle|$. In $G_{v_1, v_5}$ the vertices $v_1, v_5$ corresponding to the vectors $|v_1 \rangle, |v_5 \rangle$ play the role of the distinguished vertices and similarly in $G_{v_2, v_4}$. Let $G_{\text{Pit}}$ denote the orthogonality graph of the entire set of vectors ${G_{v_1, v_2}} \bigcup {G_{v_1, v_5}} \bigcup {G_{v_2, v_4}} \bigcup |v_3 \rangle$. 

We have that in any assignment $f: V(G_{\text{Pit}}) \rightarrow [0,1]$ for which $f(v_1), f(v_2) \in \{0,1\}$, if $f(v_1) = 1, f(v_2) = 1$, then we obtain a contradiction by the property of the extended $01$-gadget $G_{v_1,v_2}$. On the other hand, if $f(v_1) = 1, f(v_2) = 0$, then since $|v_1 \rangle \perp |v_3 \rangle$ we have $f(v_3) = 0$, and by the property of the extended $01$-gadget $G_{v_1, v_5}$ we have $f(v_5) = 0$. This gives a contradiction since $v_2, v_3, v_5$ form a maximum clique. Similarly, if $f(v_1) = 0, f(v_2) = 1$, then since $|v_2 \rangle \perp |v_3 \rangle$ we have $f(v_3) = 0$, and by the property of the extended $01$-gadget $G_{v_2, v_4}$ we have $f(v_4) = 0$. This also gives a contradiction since $v_1, v_3, v_4$ form a maximum clique. Therefore, we have any assignment $f: V(G_{\text{Pit}}) \rightarrow [0,1]$ which obeys $f(v_1), f(v_2) \in \{0,1\}$ also must necessarily obey $f(v_1) = f(v_2) = 0$. This completes the proof.

\begin{figure}
	\centerline{\includegraphics[scale=0.32]{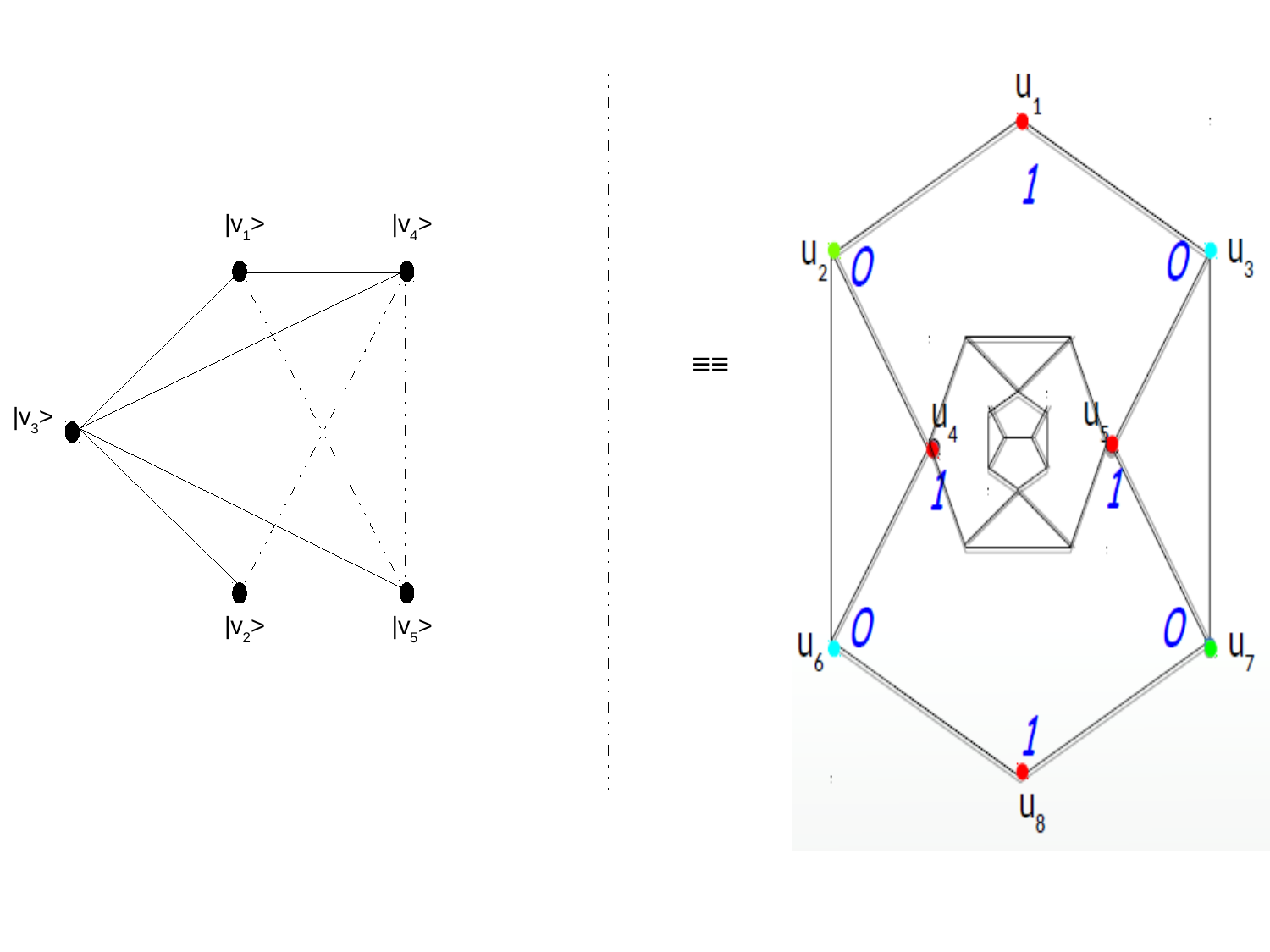}}
	\caption{A constructive proof of the extended Kochen-Specker theorem \ref{thm:HP03} using the extended $01$-gadgets. Given vectors $|v_1 \rangle, |v_2 \rangle \in \mathbb{C}^d$, we obtain vector $|v_3 \rangle \perp \text{span}(|v_1 \rangle, |v_2 \rangle)$ and two other vectors $|v_4 \rangle, |v_5 \rangle$ in the plane $\text{span}(|v_1 \rangle, |v_2 \rangle)$ with the orthogonality relations indicated in the left figure. Dashed edges between two vertices indicate an extended $01$-gadget from Theorem \ref{prop:101-gadg-const} with the two vertices being distinguished.}
	\label{fig:Ext-KS}
\end{figure} 
\end{proof}
\color{black}

\subsection{Discussion} Intuitively, with respect to any $\{0,1\}$ coloring, a $01$-gadget behaves like a ''virtual edge" between its two special vertices, with this edge also obeying the rule that at most one of its incident vertices may be assigned the color $1$. Moreover, in Theorem \ref{prop:fin-gadg-const} we have shown that $01$-gadgets may be constructed with any two non-orthogonal vectors as the special vertices. Starting from a given set of vectors, this allows us to connect any two non-orthogonal vectors by an appropriate $01$-gadget, which imposes additional constraints on the $\{0,1\}$-colorings of the resulting set of vectors. By appropriately adding such virtual edges, we are eventually able to obtain a set of vectors that gives a Kochen-Specker contradiction. Moreover, it turns out that the statistical proofs of the Kochen-Specker theorem can also be interpreted in the same manner. For instance, the famous Yu-Oh graph of \cite{YO12} can be interpreted as six $01$-gadgets connecting the vectors $(1,1,1)^T, (1,1,-1)^T, (1,-1,1)^T$ and $(-1,1,1)^T$. These four vectors thus form a ''virtual clique", with the property that in any $\{0,1\}$-coloring of the Yu-Oh set, the sum of the values attributed to these four vectors cannot exceed one. On the other hand, any quantum state has overlap with these four vectors summing to $4/3$ providing a statistical contradiction. Similar considerations also apply to the extended Kochen-Specker theorem of Pitowsky by means of extended $01$-gadgets.

\section{Computational complexity of $\{0,1\}$-colorings}\label{sec:compl}
\label{sec:Comp-complexity}

Clearly, complete graphs of size $d+1$ cannot be faithfully realized in $\mathbb{C}^d$, but there also exist certain other graphs that cannot be faithfully realized in $\mathbb{C}^d$. The well-known example is the four-cycle (square) graph in $\mathbb{C}^3$, this can be seen by the following simple argument. Suppose a pair of vertices in opposite corners of the square is assigned without loss of generality the vectors $| 0 \rangle$ and $\alpha |0 \rangle + \beta | 1 \rangle$, with $\alpha, \beta \in \mathbb{C}$. Since these vectors span a plane and the remaining pair of vertices are both required to be orthogonal to this plane, these latter vectors are both equal up to a phase to $| 2 \rangle$, contradicting the requirement of faithfulness. There exist analogous graphs that are not faithfully realizable in higher dimensions, some of which are shown in Fig. \ref{fig:forbidden-graphs}. 
\begin{figure}
	\centerline{\includegraphics[scale=0.32]{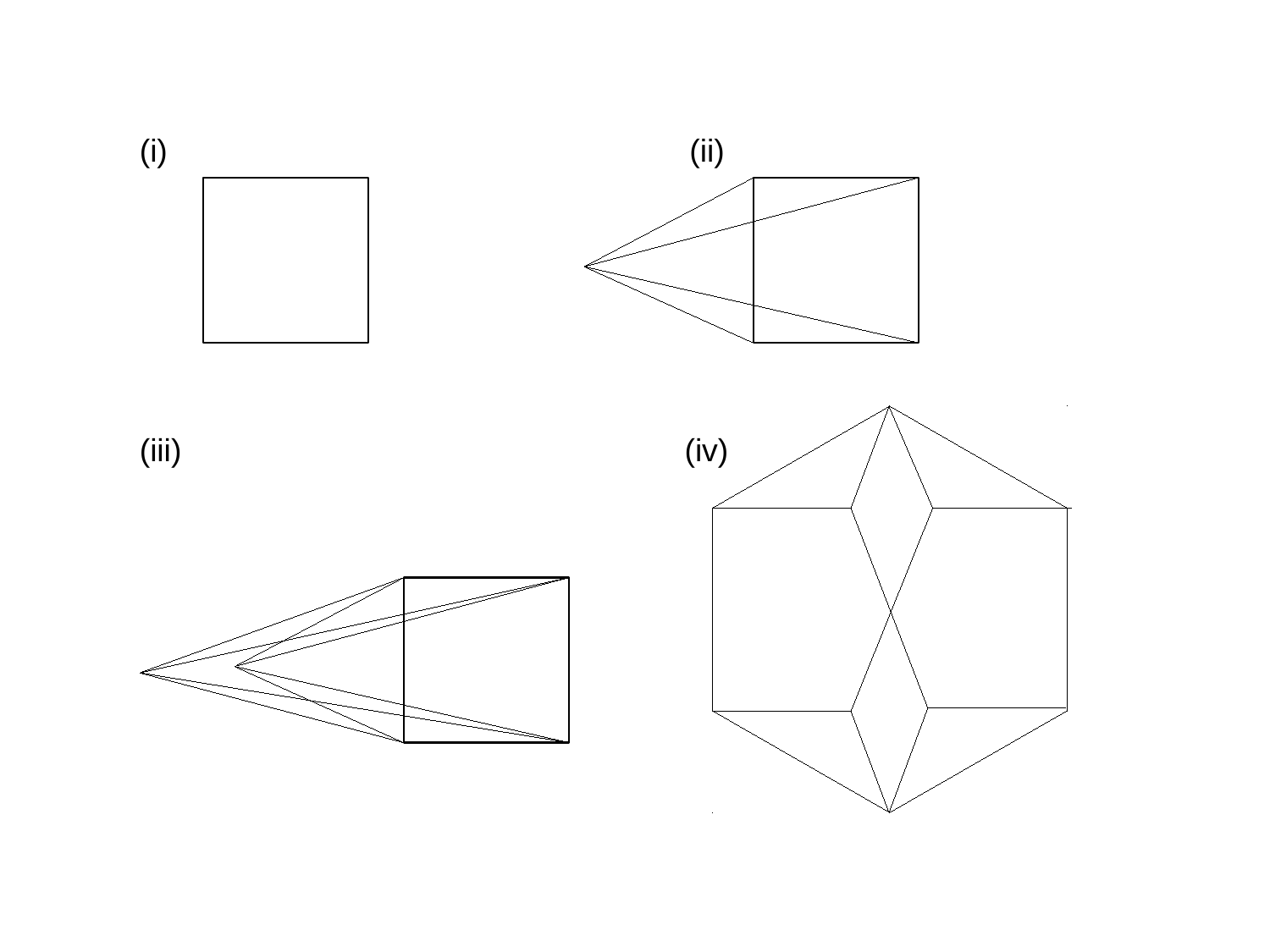}}
	\caption{Examples of forbidden subgraphs in dimensions $3,4$ and $5$. Graph (i) is the square graph which is not faithfully realizable in $\mathbb{C}^3$ as explained in the text. Graph (iv) is the graph from \cite{C11} which was verified to be not faithfully realizable in dimension three despite being square-free. Graph (ii) is not faithfully realizable in $\mathbb{C}^4$, which can be seen as arising from the fact that the induced square subgraph is not faithfully realizable in $\mathbb{C}^3$ and the additional vertex being adjacent to all vertices of the square, the vector corresponding to this vertex occupies an orthogonal subspace to that spanned by the square. Graph (iii) is similarly not realizable in $\mathbb{C}^5$ this time owing to the presence of two vertices (which themselves cannot be represented by identical vectors) that are adjacent to all the vertices of the square. It is clear that the construction can be extended to higher dimensions.}
	\label{fig:forbidden-graphs}
\end{figure} 
In searching for Kochen-Specker vector systems in $\mathbb{C}^d$, it is therefore crucial to reduce the size of the search by restricting to non-isomorphic graphs which do not contain these forbidden graphs as subgraphs. Indeed, searching over non-isomorphic square-free graphs lead to the proof that the smallest Kochen-Specker vector system in $\mathbb{C}^3$ is of size at least $18$ \cite{Arends09}.

Let us denote the set of forbidden graphs in $\mathbb{C}^d$ as $\{G_{\text{fbd}}\}$. We show, following the proof by Arends et al. \cite{Arends09, AOW11} for the square-free case, that the problem of checking $\{0,1\}$-colorability of $\{G_{\text{fbd}}\}$-free graphs is NP-complete. Here, by a $\{G_{\text{fbd}}\}$-free graph we mean a graph that does not contain any of the forbidden graphs as subgraphs. 
\begin{theorem}[see also \cite{Arends09}]\label{theo:NP}
	Checking $\{0,1\}$-colorability of $\{G_{\text{fbd}}\}$-free graphs is NP-complete. 
\end{theorem}
The proof is based on a reduction to the well-known graph coloring problem that uses $01$-gadgets in a crucial manner. Let us first recall the usual notion of coloring of a graph used in the proof. A proper coloring $c$ of a graph $G$ is an assignment of one among $n$ colors to each of the vertices of the graph $c : V(G) \rightarrow [n]$ ($[n]:=\{1,\dots,n\}$) such that no pair of adjacent vertices are assigned the same color. If such a coloring exists, we say that $G$ is $n$-colorable.

\begin{proof}
	The proof generalizes and simplifies that for the analogous question of $\{0,1\}$-colorability of square-free graphs in \cite{Arends09}, with the difference being that we directly use the constructions of $01$-gadgets from the previous sections. Firstly, we know that checking $\{0,1\}$-colorability of a $\{G_{\text{fbd}}\}$-free graph is in NP because the problem of checking an arbitrary graph for $\{0,1\}$-colorability is in NP \cite{Arends09}. Suppose we are given a graph $G$. The idea is to construct a new graph $H$ which is $\{G_{\text{fbd}}\}$-free such that the problem of $\omega(G)$-colorability of $G$ is equivalent to the problem of $\{0,1\}$-colorability of $H$. Provided the construction is achievable in polynomial time, this gives a reduction from the $\{0,1\}$-colorability problem to the $\omega(G)$-colorability problem (for $\omega(G) \geq 3$) which is known to be NP-complete \cite{CLRS01}. 
	
	The construction goes as follows. Replace every vertex $v \in V(G)$ by a clique of size $\omega(G)$ in $H$ and label the corresponding vertices $v_i \in V(H)$ for $i \in [\omega(G)]$. For every edge $(u,v) \in E(G)$, connect the corresponding vertices $(u_i,v_i)$ by a $01$-gadget $\Gamma_{(u_i,v_i)}$ in $H$. The exact form of the gadget $\Gamma_{(u_i,v_i)}$ is left unspecified at the moment, for the polynomial time reduction it is only important that it is finite (i.e., $|V(\Gamma_{(u_i,v_i)})|$ and $|E(\Gamma_{(u_i,v_i)})|$ are finite), so that $|V(H)| \leq \omega(G) (|V(\Gamma_{(u_i,v_i)})|_{max}-1) |V(G)|$ and $|E(H)| \leq \omega(G) | V(G)| + |E(\Gamma_{(u_i,v_i)})|_{max} |E(G)|$, i.e., $|V(H)| = O(|V(G)|)$ and $|E(H)| = O(|E(G)| + |V(G)|)$. 
	
	We first verify that $H$ is $\{G_{\text{fbd}}\}$-free. We do this by showing that $H$ is in fact faithfully realizable in dimension $\omega(G)$ and consequently free of the forbidden subgraphs for that dimension. For the vertices $v \in V(G)$, the actual representation of the vertices $v_i \in V(H)$ is chosen independent of the exact structure of the graph, i.e., for any $G$ with $|V(G)| = n$, we choose a fixed faithful orthogonal representation $\{|v_i \rangle\}$ for $v \in V(G)$ and $i \in [\omega(G)]$. Indeed, to show the realizability of the rest of $H$, it suffices to show the realizability of the vertices $v_1$ for $v \in V(G)$, since the representation for the remaining vertices $v_i$ for $i \geq 2$ can be readily obtained by a cyclic permutation $\Pi_i: | j \rangle \mapsto | j + i \rangle$ with the sum taken modulo $\omega(G)$. The structure of the graph is then incorporated by means of an appropriate choice of the gadgets $\Gamma_{(u_i,v_i)}$.  The crucial idea behind the construction is that there exist finite sized gadgets (with faithful representations) for any two distinct vertices as shown in Prop. \ref{prop:101-gadg-const}. So that for any edge $(u,v) \in E(G)$, we use a gadget $\Gamma(u_1,v_1)$ from Prop. \ref{prop:101-gadg-const} (the same gadget is used for the other pairs $(u_i,v_i)$) corresponding to the required overlap $|\langle u_1 | v_1 \rangle|$. Now, since the representation is faithful, we do not have different vertices represented by the same vector. As such, the construction from Prop. \ref{prop:101-gadg-const} yields a finite sized gadget for any pair of vertices $(u_i,v_1)$.   
	
	The proof that checking $\{0,1\}$-colorability of the  $\{G_{\text{fbd}}\}$-free graph $H$ is equivalent to checking the $\omega(G)$-colorability of $G$ (which is $NP$-complete) follows along analogous lines to the proof in \cite{Arends09} and we present it here for completeness. Firstly, we show that $H$ is $\{0,1\}$-colorable if $G$ is $\omega(G)$-colorable. Consider the intermediate situation when we form a graph $G'$ by replacing every vertex $v \in V(G)$ by a clique of size $\omega(G)$ and labeling the corresponding vertices $v_i \in V(G')$ for $i \in [\omega(G)]$. For every edge $(u,v) \in E(G)$, connect the corresponding vertices $(u_i,v_i)$ by an edge in $G'$. The strategy is to show that if $G$ is $\omega(G)$-colorable, then $G'$ admits a valid $\{0,1\}$-assignment. Suppose $G$ is $\omega(G)$-colorable, and $c: V(G) \rightarrow [\omega(G)]$ is an optimal coloring. We define the $\{0,1\}$-coloring of $G'$ by 
	\[
	c'(v_i) = \left\{\begin{array}{lr}
	1, & \text{for } i = c(v)\\
	0, & \text{else} 
	\end{array}\right. 
	\]
	The fact that this is a valid $\{0,1\}$-coloring of $G'$ follows the proof of Lemma 1 in \cite{Arends09}. We now derive the $\{0,1\}$-coloring of $H$ from that of $G'$ by seeing that each of the gadgets in Prop. \ref{prop:101-gadg-const} can be $\{0,1\}$-colored in all three cases, when the distinguished vertices $u_i, v_i$ have the assignments: (i) $f(u_i) = 0, f(v_i) = 0$, (ii) $f(u_i) = 0, f(v_i) = 1$, and $f(u_i) = 1, f(v_i) = 0$. This is done by checking that such a valid $\{0,1\}$-coloring exists for the Clifton gadget in Fig. \ref{fig:Clifton} in each of the three cases. The $\{0,1\}$-coloring can be extended to the entire gadget iteratively by following the procedure shown in the proof of Prop. \ref{prop:101-gadg-const}. This gives a valid $\{0,1\}$-coloring of $H$. 
	
	We now show that a valid $\{0,1\}$-coloring of $H$ also implies that $G$ is $\omega(G)$-colorable. Let $f: V(H) \rightarrow \{0,1\}$ be a valid $\{0,1\}$ assignment of $H$. For every $v \in V(G)$, by the fact that we have a valid $\{0,1\}$-coloring, exactly one of the vertices $v_i  \in V(H)$ is assigned value $1$, i.e., $f(v_i) = 1$. One can then define a $\omega(G)$-coloring $c: V(G) \rightarrow [\omega(G)]$ by $c(v) = i \leftrightarrow c(v_i) = 1$ for every $v \in V(G)$. It is clear that this is a valid coloring since if $(u,v) \in E(G)$ we have by the property of the gadget that at most one of $u_i, v_i$ is assigned value $1$, i.e., either $f(u_i) = 0$ or  $f(v_i) = 0$. Thus, the $\{0,1\}$-colorability of the $\{G_{\text{fbd}}\}$-free graph $H$ is equivalent to the $\omega(G)$-colorability of $G$. From \cite{CLRS01}, we know that for $\omega(G) \geq 3$, this problem is NP-complete, which finishes the proof. 
\end{proof}


%

It is also interesting to examine the complexity of identifying $01$-gadgets. In this case, it appears to be necessary to enumerate all $\{0,1\}$-colorings of a given graph and to check $O(n^2)$ vertices to identify the possible distinguished vertices. Note that for a graph with $n$ vertices there are $2^n$ possible $\{0,1\}$-colorings so that it is not apparent whether even a polynomially checkable certificate exists for this problem. Peeters  in \cite{P96} gave a polynomial time reduction preserving graph planarity of the problem of testing $\xi(G) \leq 3$ to the problem of testing whether the chromatic number $\chi(G)$ is less than or equal to $3$, which is a well-known NP-complete problem, so that it is hard to check whether $d(G) \leq 3$ already for the case of planar graphs.  

\section{Randomness from $01$-gadgets}
In this section, we give a brief outline of how $01$-gadgets may be linked to device-independent randomness certification.
Namely, when two parties Alice and Bob perform locally the measurements from the Clifton gadget on their half of a maximally entangled state (in $\mathbb{C}^3 \otimes \mathbb{C}^3$), we will show that some specific outcome of their joint measurements has probability bounded from above and below (and this holds in all no-signaling theories). This can be
inserted into a fully device-independent protocol as given in \cite{PRL-rand}, the details are deferred to a separate paper \cite{R17}.
To show how the Clifton gadget can be used for randomness amplification we first consider a non-contextual assignment
of probabilities to its vertices $v$ satisfying
\be
\label{eq:klik}
\sum_{v\in \text{clique}}p_v\leq 1,\quad \sum_{v\in \text{maximum clique}} p_v=1
\ee
This is the same requirement as Eq.(\ref{eq:01rule}), but we now assign not necessarily zeros and ones, but probabilities (i.e., values in $[0,1]$ rather than in $\{0,1\}$). Recall that such an assignment was also considered in our discussion of the extended Kochen-Specker theorem in Section \ref{sec:ext}.  
Now, since the gadgets are $\{0,1\}$ colorable, such an assignment of zeros and ones is possible, although in the $\{0,1\}$ assignment, it is not possible to assign $1$'s to both vertices $1$ and $8$. 
Here, we will first show, that even if we assign probabilities, we still cannot have $p_1=p_8=1$, and we will provide a quantitative bound for this.
Indeed, let us write  Eq.\eqref{eq:klik} explicitly for the cliques in question from the Clifton gadget in Fig. \ref{fig:Clifton}:
\ben
\label{eq:nonmax}
&&p_1+p_2\leq 1, \quad p_1+p_6\leq 1, \quad p_4+p_5\leq 1, \nonumber \\
&&\quad p_7+p_8\leq 1, \quad p_3+p_2\leq 1
\een
for non-maximal cliques and 
\be
\label{eq:max}
p_2+p_3+p_4=1, \quad  p_5+p_6+p_7=1
\ee
for the two maximum cliques. 
We sum up all the inequalities \eqref{eq:nonmax}, and get 
\be
2p_1 + p_2+p_3+p_4+p_5+p_6+p_7+ 2 p_8 \leq 5.
\ee
Using \eqref{eq:max} we then obtain
\be
\label{eq:nsbound}
p_1+p_8\leq \frac32.
\ee
To exploit this feature for randomness amplification, 
we consider a maximally entangled state shared by two parties. The parties will measure observables composed of the 
projectors given by the quantum representation (if the clique is not maximal, one simply adds a third orthogonal projector to obtain a complete measurement).
Recall here that a set of eight projectors $P_v=|u_v\>\<u_v|$ that is compatible with the Clifton graph
is given in Eq.\eqref{eq:Clif-orth-rep}.
Projectors of Alice will be denoted $A_v$ and those of Bob $B_v$, 
and the probability of obtaining outcome $v$, while measuring observable containing $v$, will be denoted by $p(A_v=1)$.
We correspondingly denote by  $p(A_v=0)$ the probability that the outcome $v$ was not obtained. Clearly $p(A_v=1)+p(A_v=0)=1$.
Now, we shall show using no-signaling (which will impose non-contextuality), that the probability $p(A_1=1,B_8=1)$ 
is bounded from above. 
To see this, we apply Eq.\eqref{eq:nsbound} to Alice's observables and get
\ben
&&p(A_1=1)+p(A_8=1)\leq \frac32
\een
From the correlations in the maximally entangled state, we have that 
\be
p(A_8=1)=p(B_8=1)
\ee
giving 
\be
\label{eq:AB32}
p(A_1 = 1)+p(B_8 = 1)\leq \frac32.
\ee
Now, from no-signaling we have 
\ben
&&p(A_1=1)=  p(A_1 = 1, B_8 = 0) + p(A_1 = 1, B_8 = 1), \nonumber \\
&&p(B_8 = 1) =  p(A_1 = 0, B_8 = 1) + p(A_1 = 1, B_8 = 1).  \nonumber\\
\een
Summing these and applying \eqref{eq:AB32} we get 
\ben
&&p(A_1 = 1, B_8 = 0) + p(A_1 = 1, B_8 = 1)   +  \nonumber \\
&& p(A_1 = 0, B_8 = 1) + p(A_1 = 1, B_8 = 1) \leq   \frac32 
\een
and hence 
\be
p(A_1=1,B_8=1)\leq \frac34.
\ee
Thus we have obtained, that the probability of the event $(A_1,B_8)=(1,1)$ is bounded from above. 
We have also the lower bound
\be
p(A_1=1,B_8=1) = \frac13  |\<u_1|u_8\>|^2 \geq \frac1{27}.
\ee
Thus 
\be
\frac1{27}\leq p(A_1=1,B_8=1) \leq \frac34
\ee
Therefore, the outcome $(A_1,B_8)=(1,1)$ has randomness, which can be used in a randomness amplification scheme employing the protocol of \cite{PRL-rand}.
The lower bound is $\frac1{27}$ in noiseless conditions, and assuming we have exactly measured the specified projectors. 
In a real experiment, this value may be different, but if the noise is low enough it should be close to $\frac1{27}$.
Also the upper bound, relies on perfect correlations, which in a real experiment may be imperfect. Thus in noisy conditions, we will 
have less stringent lower and upper bounds, though these are certifiable by statistics from the experiment. 
Note that crucially we have not used explicitly Bell inequalities, nor even the KS paradox. We have simply made use of the perfect correlations between the parties and the local  $01$-gadget structure of Alice and Bob's observables. \\

\section{Conclusion and Open Questions}\label{sec:concl}
In this paper, we have shown that there exist interesting subgraphs of the Kochen-Specker graphs that we termed $01$-gadgets that encapsulate the main contradiction necessary to prove the  Kochen-Specker theorem. Furthermore, as a main technical contribution, we have shown that the fundamental structures identified here, lead to clean constructions of state-independent statistical proofs of the KS theorem, of which the famous Yu and Oh proof is a particular case. The proofs given here provide a new perspective on these results, and serve as a useful tool to construct minimal KS sets, since efforts may be concentrated on the $01$-gadget subgraphs. 
An extended notion of $01$-gadgets also helped to provide simple constructive proofs of the extended Kochen-Specker theorem \cite{Pitowsky}. The gadgets enable a proof of the NP-completeness of checking $\{0,1\}$-colorability of graphs free from the forbidden subgraphs from Hilbert spaces of any dimension. Practically, the gadgets open up a highly important application of contextuality to practical device-independent randomness generation, which we study in a companion paper \cite{R17} where we provide an explicit device-independent protocol for randomness amplification based on \cite{BRGH+16, PRL-rand, WBGH+16} and Hardy paradoxes constructed using $01$-gadgets. 


An open question, is to find, for given overlap  $| \langle v_1 | v_2 \rangle|$, the minimal $01$-gadget and extended $01$-gadget with the corresponding vertices $v_1,v_2$ playing the role of the distinguished vertices. An answer to this question would have applications for randomness generation from contextuality \cite{R17}. 
Another open question is whether all state-independent contextual graphs (including those going beyond KS sets such as that of Yu and Oh \cite{YO12}) contain $01$-gadgets as subgraphs, or even possibly as induced subgraphs. Finally, while it is known that in $\mathbb{C}^3$ KS sets cannot be constructed using rational vectors \cite{GN08}, it would be very interesting to study quantum realizations of $01$-gadgets using rational vectors, to build statistical KS arguments and state-independent non-contextuality inequalities using these.  


\textit{Acknowledgments.-}
We are grateful to Andrzej Grudka, Waldemar K\l obus and David Roberson for useful discussions.
R.R. acknowledges support from the research project  ``Causality in quantum theory: foundations and applications'' of the Fondation Wiener-Anspach and from the Interuniversity Attraction Poles 5 program of the Belgian Science Policy Office under the grant IAP P7-35 photonics@be. This work is supported by the Start-up Fund 'Device-Independent Quantum Communication Networks' from The University of Hong Kong. This work was supported by the National Natural Science Foundation of China through grant 11675136, the Hong Kong Research Grant Council through grant 17300918, and the John Templeton Foundation through grants 60609, Quantum Causal Structures, and 61466, The Quantum Information Structure of Spacetime (qiss.fr). M. R. is supported by the National Science Centre, Poland, grant OPUS 9. 2015/17/B/ST2/01945. MH and PH are supported by the John Templeton Foundation. The opinions expressed in this publication are those of the authors and do not necessarily reflect the views of the John Templeton Foundation. PH also acknowledges partial support from the Foundation for Polish Science (IRAP project, ICTQT, contract no. 2018/MAB/5, co-financed by EU within the Smart Growth Operational Programme). KH acknowledges support from the grant Sonata Bis 5 (grant number: 2015/18/E/ST2/00327) from the National Science Centre. S.P. is a Research Associate of the Fonds de la Recherche Scientifique (F.R.S.-FNRS). We acknowledge support from the EU Quantum Flagship project QRANGE.

\end{document}